\documentclass{article}
\usepackage{amssymb}
\usepackage{amsfonts}
\usepackage{amsmath,mathtools}
\usepackage{graphicx}
\usepackage{amsbsy}
\usepackage{bm}
\usepackage{float}
\usepackage[utf8]{inputenc}
\usepackage[english]{babel}
\usepackage{authblk}
\usepackage{natbib}
\usepackage{color}
\usepackage{amsthm}
\usepackage{subcaption}
\usepackage{longtable}
\usepackage[hidelinks]{hyperref}
\usepackage{tipa, float}
\hypersetup{colorlinks = true, urlcolor = blue, linkcolor = blue, citecolor = red}
\newcommand{\pkg}[1]{{\normalfont\fontseries{b}\selectfont #1}}
\let\proglang=\textsf
\usepackage{multirow}
\usepackage[final]{changes}
\definechangesauthor[name={RM}, color=orange]{rm}

\newcommand{\diag}{\textrm{diag}}

\newtheorem{lemma} {Lemma}[section]

\newtheorem{theorem} {Theorem}[section]

\newtheorem{corollary}{Corollary}[section]

\DeclareMathAlphabet\mathbfcal{OMS}{cmsy}{b}{n} 
\setcitestyle{authoryear} 

\usepackage{multicol}
\usepackage{graphicx}

\begin{document}
	\title{Regularized Multivariate Functional Principal Component Analysis}
	\author{Hossein Haghbin\footnote{Artificial Intelligence and Data Mining Research Group, ICT Research Institute, Faculty of Intelligent Systems Engineering and Data Science, Persian Gulf University, Boushehr, Iran.}\ , Yue Zhao$^\dag$, and Mehdi Maadooliat\footnote{Department of Mathematical and Statistical Sciences, Marquette University, Milwaukee, WI 53233.}}
	\maketitle
\begin{abstract}
Multivariate Functional Principal Component Analysis (MFPCA) is a valuable tool for exploring relationships and identifying shared patterns of variation in multivariate functional data. However, controlling the roughness of the extracted Principal Components (PCs) can be challenging. This paper introduces a novel approach called regularized MFPCA (ReMFPCA) to address this issue and enhance the smoothness and interpretability of the multivariate functional PCs. ReMFPCA incorporates a roughness penalty within a penalized framework, using a parameter vector to regulate the smoothness of each functional variable. The proposed method generates smoothed multivariate functional PCs, providing a concise and interpretable representation of the data. Extensive simulations and real data examples demonstrate the effectiveness of ReMFPCA and its superiority over alternative methods. The proposed approach opens new avenues for analyzing and uncovering relationships in complex multivariate functional datasets.
\end{abstract}  
	
\section{Introduction}
In recent years, there has been a notable surge in interest towards functional data analysis (FDA) owing to its remarkable capability to effectively handle intricate data structures, including longitudinal data, image data, and time series data. Multivariate functional data, characterized by observations composed of vectors of functions, holds significant relevance across diverse applications such as biomedical imaging, environmental science, and finance.

Functional Principal Component Analysis (FPCA) is a key technique in FDA for dimension reduction by capturing the primary modes of variation in functional variables, known as functional principal components (PCs). The foundational contributions to FPCA for univariate functional data can be attributed to the works of \cite{rice1991estimating} and \cite{pezzulli1993some}. For a comprehensive survey of the topic, see \cite{shang2014survey}.	

To mitigate the issue of undesirable fluctuations in functional PCs extracted from (univariate) functional data, various regularization techniques have been proposed \citep{silverman1996smoothed, ramsay2005, huang2008functional, qi2011some,lakraj2017some}. These methods aim to enhance the interpretability and stability of the PCs, ultimately improving the accuracy of the results. Throughout this paper, we refer to \citeauthor{silverman1996smoothed}'s \citeyear{silverman1996smoothed} approach as regularized FPCA (ReFPCA).
	
Multivariate FPCA (MFPCA) extends the concept of FPCA to the multivariate setting, where each observation comprises a vector of functions. MFPCA enables the exploration of relationships between multiple functional variables and the identification of shared patterns of variation. Numerous approaches have been proposed for MFPCA in the literature. \cite{ramsay2005} provides a concise overview of MFPCA for bivariate functional data, while \cite{jacques2014model} apply MFPCA in the context of clustering. \cite{chiou2014multivariate} introduce a normalized MFPCA method that yields a normalized Karhunen-Loeve representation of the data.  \cite{happ2018multivariate} propose a generalized MFPCA framework applicable to functional data across various domains with possibly different dimensions.

Despite the recognized advantages of ReFPCA for univariate functional data, there is a research gap regarding the definition of an appropriate smoothness level for MFPCA. To address this gap, we propose a novel approach called regularized MFPCA (ReMFPCA), which employs a parameter vector to control the smoothness of each functional variable. By incorporating smoothness constraints as penalty terms within a regularized optimization framework, ReMFPCA generates smoothed multivariate functional PCs, leading to a concise and interpretable representation of the multivariate functional data. We validate the effectiveness of ReMFPCA through simulations and real data examples, comparing its performance with that of MFPCA, marginal FPCA and ReFPCA methods.

The paper is organized as follows. Section \ref{notation} introduces the preliminary notations, while Section \ref{MFPCA} offers a concise review of the relevant literature. The proposed ReMFPCA method and its properties are presented in Section \ref{ReMFPCA}, followed by implementation details in Section \ref{implement}. To demonstrate the effectiveness of the proposed method, Sections \ref{simul} and \ref{application} provide simulation results and a real data example, respectively. Lastly, Section \ref{conclusion} concludes the paper, highlighting key findings, and suggests potential avenues for future research.

\section{Preliminary Notations}\label{notation}
We consider $\mathcal{T}_j,\ j=1,\ldots, p,$ as a compact subset of $\mathbb{R}$, where $p$ denotes the number of variables. Let's define $H_j:=L^2(\mathcal{T}_j)$, consisting of measurable functions $x(\cdot)$ on $\mathcal{T}_j$ with $\int_{\mathcal{T}_j} |x(t)|^2 dt < \infty$. $H_j$ is a Hilbert space equipped with the inner product $\langle x , y \rangle_{H_j}=\int_{\mathcal{T}_j} x(t)y(t) dt$ and the norm $\Vert x \Vert_{H_j}=\sqrt{\langle x , x \rangle_{H_j}}$. The tensor (outer) product $x\otimes y$ corresponds to the operator $x\otimes y: 
{H}_i\rightarrow {H}_j$, defined as ($x\otimes y)h:=\langle x, h \rangle y$, where $h\in {H}_i$.

Functional random variables is represented by univariate random processes $X_j=\{X_j(t): \ t\in\mathcal{T}_j\}$, with mean function $\mu_j(\cdot)$ and finite covariance function $c_j(s,t)=\text{cov}\left(X_j(s),X_j(t)\right)$. The space $L^2_{H_j}$ refers to the space of $H_j$-valued random variables $X_j$ satisfying $E(\Vert X_j\Vert_{H_j}^2)<\infty$. This space is a Hilbert space equipped  with the inner product $\langle X,Y\rangle_{L^2_{H_j}}=E\langle X,Y\rangle_{{H}_j}, \ X, Y\in L^2_{H_j}$. Let $c_{i,j}(s,t)=\text{cov}\left(X_i(s),X_j(t)\right)$ for $X_i\in L^2_{H_i}$ and $X_j\in L^2_{H_j}$, the associated cross-covariance operator is an integral operator with respect to the kernel function $c_{i,j}(\cdot,\cdot)$ i.e.,
		\begin{equation*}
			\mathcal{C}_{i,j}(f):=\int_{\mathcal{T}_j} c_{i,j}(t,\cdot)f(t) dt,\ f \in H_i.
		\end{equation*}  
The cartesian product space $\mathbb{H}=H_1\times H_2\times \cdots \times H_p$ is a Hilbert space equipped with the inner product $\langle {\pmb x,\pmb y} \rangle_\mathbb{H}=\sum_{j=1}^p\langle  x_j,y_j \rangle_{H_j}$. The data of interest in this paper consists of multivariate functional random variables ${\pmb X}=\left( X_1, X_2, \ldots, X_p\right)$, belonging to $L^2_\mathbb{H}$, with a domain $\mathbfcal{T}=\mathcal{T}_1\times \cdots \times\mathcal{T}_p$. The covariance operator of ${\pmb X}$ will de defined by  $\mathbfcal{C}: \mathbb{H}\rightarrow \mathbb{H}$ as 
		\begin{align*}
			\mathbfcal{C}(\pmb f)(\pmb t):=
			\begin{bmatrix}
				\sum_{j=1}^p\int_{\mathcal{T}_j} c_{1,j}(s_j, t_1)f_j(s_j)ds_j\\	
				\vdots\\
				\sum_{j=1}^p\int_{\mathcal{T}_j} c_{p,j}(s_j, t_p)f_j(s_j)ds_j
			\end{bmatrix},
		\end{align*}
where $\pmb{t}=(t_1, \ldots, t_p)\in\mathbfcal{T}$ and $\pmb{f}=(f_1, \ldots, f_p)\in\mathbb{H}.$ This definition implies that $\mathbfcal{C}$ is linear, self-adjoint and positive operator \citep{happ2018multivariate}.

The Sobolev space $W^2_j$ is defined as
		\begin{equation*}
			W^2_j:=\left\{
			x(.): x\ and\ x'\ are\ absolutely\ continuous\ on\ \mathcal{T}_j\ and\ 
			x''\in H_j
			\right\},
		\end{equation*}	
where $x'$ and $x''$ denote the first and second derivatives of $x$, respectively. We introduce the bilinear form $[x, y]_{H_j} := \langle x'', y'' \rangle_{H_j}$  for elements $x, y \in W^2_j$. When $g'''' \in H_j$, it can be demonstrated that $[x, y]_{H_j} = \langle x, y'''' \rangle_{H_j}$\citep{silverman1996smoothed}. Given a smoothing parameter $\alpha_j > 0$, we define the inner product $\langle x, y \rangle_{\alpha_j}$ as
		$$\langle { x, y} \rangle_{\alpha_j}:=\langle { x, y} \rangle_{H_j}+\alpha_j[ { x, y} ]_{H_j}.$$ 
The corresponding norm associated with this inner product is $\Vert x \Vert_{\alpha_j} := \sqrt{\langle x, y \rangle_{\alpha_j}}$. The cartesian product of Sobolev spaces is denoted as $\mathbb{W}^2=W^2_1\times W^2_2\times \cdots \times W^2_p$. Given a smoothing parameter vector $\pmb{\alpha}=(\alpha_1,\alpha_2, \ldots, \alpha_p) \in \mathbb{R}^{p+}$, the inner product $\langle {\pmb x,\pmb y} \rangle_{\pmb \alpha}$ is defined as
		\begin{equation*}
			\langle {\pmb x,\pmb y} \rangle_{\pmb \alpha}:=\sum_{j=1}^p\langle { x_j, y_j} \rangle_{\alpha_j}=	\langle {\pmb x,\pmb y} \rangle_\mathbb{H}+\sum_{j=1}^p\alpha_j[ { x_j, y_j} ]_{H_j},
		\end{equation*}
where ${\pmb x ,\pmb y}\in \mathbb{W}^2.$ The corresponding norm is denoted by $\Vert {\pmb x} \Vert_{ \pmb \alpha}:=\sqrt{\sum_{j=1}^p\Vert {x_j} \Vert_{ \alpha_j}^2}.$	
	
\section{Multivariate functional PCA} \label{MFPCA}
In this paper we consider ${\pmb X}_1, \ldots, {\pmb X}_n \in L^2_\mathbb{H}$ are independent and identically distributed (iid) observations defined on the set $\mathbfcal{T}$. We aim to address the standard FPCA problem, which seeks a linear functional $\ell$ by maximizing the variance of $\ell(\pmb{X})$, subject to $\Vert \ell\Vert_{\mathbb{H}^*} = 1$, where $\mathbb{H}^*$ denotes the dual space of bounded linear functionals on $\mathbb{H}$. Applying the Riesz representation theorem, there is a unique $\pmb{\varphi}\in \mathbb{H}$ associated with any $\ell\in\mathbb{H}^*$, such that $\ell(\pmb{X})=\langle \pmb{X} , \pmb{\varphi} \rangle_{\mathbb{H}}$ and $\Vert \ell\Vert_{\mathbb{H}^*} = \Vert \pmb{\varphi}\Vert_{\mathbb{H}}$. Consequently, the variance of $\ell(\pmb{X})$ can be expressed as $\mathbb{E}\left(\langle \pmb{X-\mu} , \pmb{\varphi} \rangle_{\mathbb{H}}^2\right)= \langle \mathbfcal{C}\pmb{\varphi} , \pmb{\varphi} \rangle_{\mathbb{H}}$.

The MFPCA problem involves finding the first multivariate functional PC, denoted as $\pmb{\psi}_1\in {\mathbb{H}}$, which solves the optimization problem:
		\begin{equation}
			\label{mfpca1}
			\max_{\pmb{\varphi}:\Vert\pmb{\varphi}\Vert_{\mathbb{H}}=1}{\dfrac{\langle 
			\mathbfcal{C}\pmb{\varphi} , \pmb{\varphi} \rangle_{\mathbb{H}}}{\Vert 
			\pmb{\varphi} \Vert_{\mathbb{H}}^2}}.
		\end{equation}
Furthermore for $j>1$, the $j$-th multivariate functional PC, denoted as $\pmb{\psi}_j\in {\mathbb{H}}$, is obtained by solving the following optimization problem:
		\begin{equation}
			\label{mfpca2}
			\max_{\pmb{\varphi}:\Vert\pmb{\varphi}\Vert_{\mathbb{H}}=1, \langle 
			\pmb{\varphi} , \pmb{\varphi}_i \rangle_{\mathbb{H}}=0,\ i=1,\ldots, j-1} 
			{\dfrac{\langle \mathbfcal{C}\pmb{\varphi} , \pmb{\varphi} 
			\rangle_{\mathbb{H}}}{\Vert \pmb{\varphi} \Vert_{\mathbb{H}}^2}}.
		\end{equation}
The solutions to the successive optimization problems \eqref{mfpca1} and \eqref{mfpca2} correspond to the leading solutions of the eigenfunction equation, ${\mathbfcal{C}\pmb{\varphi}} =\lambda  \pmb{\varphi}$. Since the covariance operator $\mathbfcal{C}$ is typically unknown, direct estimation of the population multivariate functional PCs is not feasible. Instead, researchers rely on the sample covariance function ${\widehat{\mathbfcal{C}}}$ to estimate $\mathbfcal{C}$ and use the eigenvalues and eigenfunctions of ${{\widehat{\mathbfcal{C}}}}$ as estimates of those of $\mathbfcal{C}$. Specifically consider a sample $\pmb{x}_1,\ldots,\pmb x_n$ of observations from $\pmb{X}_1,\ldots,\pmb X_n$, where $\pmb{x}_i=(x_{i,1},\ldots, x_{i,p})^\top$,  without loss of generality we  assume that the $\pmb{x}_i$'s are centered (i.e., their sample mean is zero). The sample estimator of the covariance operator $\mathbfcal{C}$ is given by:
		\begin{equation*}
			{\widehat{\mathbfcal{C}}}=\dfrac{1}{n-1}\sum_{i=1}^n 
			\pmb{x}_i\pmb{\otimes}\pmb{x}_i.
		\end{equation*}

It has been demonstrated that the number of sample PCs with positive eigenvalues, denoted as $k$, satisfies $k\leq n$, where $n$ is the sample size. The collection of sample multivariate functional PCs is denoted as $\lbrace \hat{\pmb\psi}_i: i=1,\ldots, k \rbrace$, and the corresponding estimated eigenvalues is denoted as $\lbrace\hat{\lambda}_i: i=1,\ldots, k \rbrace$.

\section{Regularized Formulation}\label{ReMFPCA}
The ReFPCA method, proposed by \cite{silverman1996smoothed} (see  \cite{ramsay2005} Chapter 9 for more details), introduces smoothing by modifying the norm used in traditional FPCA. To develop a similar approach for the multivariate case, we define the regularized multivariate functional PC as the solution to the following optimization problem:
		\begin{equation}
			\label{e0}
			\max_{\pmb{\varphi}:\Vert\pmb{\varphi}\Vert_{\pmb\alpha}=1}\dfrac{\langle 
			\mathbfcal{C}\pmb{\varphi} , \pmb{\varphi} \rangle_{\mathbb{H}}}{
				\sum_{j=1}^p\left( \Vert {\varphi_j } \Vert_{H_j}^2+ \alpha_j 
				[\varphi_j]_{H_j}^2\right)},
		\end{equation}
where $\mathcal{D}^2$ is the second-derivative operator, and $\pmb\alpha$ is a positive vector of smoothing parameters.

Let ${\lambda}_1^{[\pmb\alpha]}$ and ${\pmb\psi}_1^{[\pmb\alpha]}$ represent the maximum value and solution, respectively, obtained from the optimization problem \eqref{e0}. Now, for any $k\in\mathbb{N}$, assuming we have obtained $\lbrace ({\lambda}_i^{[\pmb\alpha]}, {\pmb\psi}_i^{[\pmb\alpha]}): i=1,\ldots, k-1 \rbrace$, the solution ${\pmb\psi}_k^{[\pmb\alpha]}$ is found by maximizing the following optimization problem:
		\begin{equation}
			\label{e0.5}
			\max_{\pmb{\varphi}:\Vert\pmb{\varphi}\Vert_{\pmb\alpha}=1, \langle 
			\pmb{\varphi},{\pmb\varphi}_i^{[\pmb\alpha]} 
			\rangle_{\pmb\alpha}=0}\dfrac{\langle {{\mathbfcal{C}}}\pmb{\varphi} , 
			\pmb{\varphi} \rangle_{\mathbb{H}}}{\Vert \pmb{\varphi} \Vert_{\pmb \alpha}^2},\ 
			\ i=1, \ldots, k-1.
		\end{equation}
We denote the maximum value of \eqref{e0.5} as ${\lambda}_k^{[\pmb\alpha]}$. Replacing ${\mathbfcal{{C}}}$ with ${\widehat{\mathbfcal{C}}}$, for $k\leq n$, the set of sample regularized multivariate functional PCs is defined as $\lbrace \hat{\pmb\psi}_i^{[\pmb\alpha]}: i=1,\ldots, k \rbrace$, and the corresponding estimated eigenvalues (which are positive) are denoted as $\lbrace\hat{\lambda}_i^{[\pmb\alpha]}: i=1,\ldots, k \rbrace$.

\subsection{Existence of Solutions}
We investigate the existence of solutions to the optimization problem \eqref{e0.5} based on the following assumptions:
\begin{itemize}
	\item[A1.] The functions $c_{i,j}$ satisfy the following conditions for all $i, j=1,\ldots, p$: 
	\begin{equation*}
		\int_{\mathcal{T}_i} c_{i,j}(s_i,t_j)ds_i < K_{ij},\quad \forall\ t_j\in\mathcal{T}_j,
	\end{equation*}
	where $K_{ij}$ is a finite constant. Additionally, the function $c_{i,j}$ is uniformly continuous:
	\begin{equation*}
		\forall\ \varepsilon>0,\ \exists\ \delta_{ij}>0: |t_j-t_j^*|< \delta_{ij} \Rightarrow |c_{i,j}(s_i, t_j)-c_{i,j}(s_i,t^*_j )| < \varepsilon,\ \forall\ s_i\in \mathcal{T}_i.
	\end{equation*}
	These assumptions ensure that the covariance operator $\mathbfcal C$ is a compact positive operator on $\mathbb{H}$, allowing the existence of a complete sequence of eigenfunctions $\pmb{\psi}_i$s with eigenvalues $\lambda_1>\lambda_2>\cdots>0$. \citep[See Proposition 2 in][]{happ2018multivariate}.
	\item[A2.] Each eigenfunction $\pmb{\psi}_i$ belongs to the space $\mathbb{W}$, implying finite roughness (i.e., $\Vert \pmb{\psi}_i \Vert_{\pmb \alpha}^2$ is finite).
	\item[A3.] All eigenvalues $\lambda_i$ have a multiplicity of 1, yielding an ordering of $\lambda_1>\lambda_2>\cdots>0$.
\end{itemize}

\begin{theorem}\label{theorem-existence}
Under assumptions A1-A3, the optimization problems \eqref{e0} and \eqref{e0.5} have almost sure solutions $\lbrace ({\lambda}_i^{[\pmb\alpha]}, {\pmb\psi}_i^{[\pmb\alpha]}): i\in\mathbb{N} \rbrace$ for any $\pmb \alpha=(\alpha_1, \ldots, \alpha_p)$, where $\alpha_i\ge 0$.
\end{theorem}
We provide the proofs of all the theorems and lemmas in the appendix for the sake of brevity and to maintain the flow of the main text. To facilitate the understanding of the subsequent discussion, we introduce the second-derivative operator $\mathcal{R}$ and the fourth-derivative operator $\mathcal{Q}:=\mathcal{R}^*\mathcal{R}$. In particular, we define the operator $S_{\alpha_j}$ as follows:
		\begin{equation*}
			\mathcal{S}_{\alpha_j} = \left( \mathcal{I}+\alpha_j \mathcal{Q}\right)^{-1/2},\quad \alpha_j>0,\ j=1,\ldots, p.
		\end{equation*}
The operator $\mathcal{S}_{\alpha_j}$ is a bounded linear self-adjoint operator mapping on ${H}_j$, with its range contained in $W_j^2$. Notably, the linear norm of $\mathcal{S}_{\alpha_j}$ does not exceed 1. This operator is commonly referred to as the ``half-smoothing'' operator and plays a crucial role in the context of this study. For more comprehensive insights, please refer to \cite{silverman1996smoothed} and \cite{qi2011some}. Moreover, it is worth mentioning that the inverse of $\mathcal{S}_{\alpha_j}$ exists and is self-adjoint.
	
Note that $\mathcal{S}_{\alpha_j}^2 f$ solves the differential equation
	\begin{equation*}
		g(x)+\alpha_j g^{''''}(x) = f(x) ,\quad \alpha_j>0,
	\end{equation*}
and $\mathcal{S}_{\alpha_j}^2$ is a bounded and positive self-adjoint operator. To introduce the multivariate counterpart of the half-smoothing operator, we define the operator $\mathbfcal{S}_{\pmb{\alpha}}: \mathbb{H}\rightarrow\mathbb{H}$ as
	\begin{equation*}\label{s_operator}
		\mathbfcal{S}_{\pmb{\alpha}} :=
		\begin{bmatrix}
			\mathcal{S}_{\alpha_1}&0&\ldots&0\\
			0&\mathcal{S}_{\alpha_2}&\ddots&\vdots\\
			\vdots&\ddots&\ddots&0\\
			0&\ldots&0&\mathcal{S}_{\alpha_p}
		\end{bmatrix}.
	\end{equation*}
It is worth noting that $\mathbfcal{S}_{\pmb{\alpha}}^k=\diag\left(\mathcal{S}_{\alpha_1}^k, \ldots, \mathcal{S}_{\alpha_p}^k\right)$ for $k=-2, -1, 2$. The following lemma provides justification for referring to $\mathbfcal{S}_{\pmb{\alpha}}$ as half-smoothing operators.

\begin{lemma}\label{half-smoothing}
	The operator $\mathbfcal{S}_{\pmb{\alpha}}$ is a bounded linear self-adjoint operator with a linear norm of at most 1. Its range, or the domain of $\mathbfcal{S}_{\pmb{\alpha}}^{-1}$, is $\mathbb{W}^2$. Additionally, for any ${\pmb{x, y}}\in\mathbb{W}^2$, we have:
\begin{equation*}
	\langle\mathbfcal{S}_{\pmb{\alpha}}^{-1}\pmb{x}, \mathbfcal{S}_{\pmb{\alpha}}^{-1}\pmb{y}\rangle_{\mathbb{H}}^2=\langle{\pmb{x ,y}}\rangle_{\pmb{\alpha}}^2.
\end{equation*}
\end{lemma}	
As an immediate consequence of Lemma \ref{half-smoothing}, we have $\Vert\mathbfcal{S}_{\pmb{\alpha}}^{-1}\pmb{x}\Vert_{\mathbb{H}}^2=\Vert{\pmb{x}}\Vert_{\pmb{\alpha}}^2$ for any ${\pmb{x}}\in\mathbb{W}^2$.

\begin{lemma}\label{rmfpca_obtain}
	Let $\lbrace ({\lambda}_i^{[\pmb\alpha]},  {\pmb\psi}_i^{[\pmb\alpha]}): i\in\mathbb{N} \rbrace$ be the solution of the successive optimization problems \eqref{e0} and \eqref{e0.5}. Then, the pairs $\lbrace ({\lambda}_i^{[\pmb\alpha]},\mathbfcal{S}_{\pmb{\alpha}}^{-1}  {\pmb\psi}_i^{[\pmb\alpha]}): i\in\mathbb{N} \rbrace$ represent the leading eigenvalues and eigenfunctions of the compact operator $\mathbfcal{S}_{\pmb{\alpha}}{{\mathbfcal{C}}}\mathbfcal{S}_{\pmb{\alpha}}$.
\end{lemma}

\subsection{Consistency and Asymptotic Results}
This subsection establishes the consistency of the proposed method in an asymptotic framework as the number of observations, denoted as $n$, increases. Let us define ${\mathbfcal P}_{\pmb\rho}$ as the projection onto the subspace generated by ${\pmb\rho}\in\mathbb{H}$:
	\begin{equation*}
		{\mathbfcal P}_{\pmb\rho}\pmb{x}=\frac{\langle\pmb\rho,\pmb{x} \rangle_{\mathbb{H}}\pmb\rho}{\Vert\pmb\rho \Vert_{\mathbb{H}}^2}=\frac{\pmb\rho\otimes\pmb\rho(\pmb{x})}{\Vert\pmb\rho \Vert_{\mathbb{H}}^2},\quad{\pmb x}\in\mathbb{H}.
	\end{equation*}

\begin{theorem}\label{theorem_cons}
Let $\lbrace (\hat{\lambda}_i^{[\pmb\alpha]}, \hat{\pmb\psi}_i^{[\pmb\alpha]}): i\in\mathbb{N} \rbrace$ be the solution of the successive optimization problems \eqref{e0} and \eqref{e0.5} where ${\mathbfcal{{C}}}$ is replaced with ${\widehat{\mathbfcal{C}}}$. We define
\begin{equation*}
	\tilde{\pmb\psi}_i^{[\pmb\alpha]}
	=
	\dfrac{\hat{\pmb\psi}_i^{[\pmb\alpha]}}{\Vert \hat{\pmb\psi}_i^{[\pmb\alpha]} \Vert_{\mathbb{H}}}.
\end{equation*} 
Under the assumptions A1-3 and $\alpha_i\to 0$, for each $i$:
\begin{equation*}\label{conv1}
	\hat{\lambda}_i^{[\pmb\alpha]} \to {\lambda}_i,
\end{equation*}	
and
\begin{equation*}\label{conv2}
	\langle \tilde{\pmb\psi}_i^{[\pmb\alpha]} , {\pmb\psi}_i \rangle_{\mathbb{H}} \to 1.
\end{equation*}
with probability 1 as $n\to\infty$.
\end{theorem}

\begin{corollary}\label{corollary_31}
	Under the assumptions of Theorem \ref{theorem_cons}, it is possible to choose the sign of $\tilde{\pmb\psi}_i^{[\pmb\alpha]}$ such that $\Vert \tilde{\pmb\psi}_i^{[\pmb\alpha]} - {\pmb\psi}_i \Vert_{\mathbb{H}} \to 0$ with probability 1 as $n\to\infty$.
\end{corollary}

\begin{corollary}\label{corollary_32}
	Under the assumptions of Theorem \ref{theorem_cons}, we have
	\begin{equation*}\label{conv3}
		\Vert {\mathbfcal P}_{\hat{\pmb\psi}_i^{[\pmb\alpha]} }- {\mathbfcal P}_{{\pmb\psi}_i} \Vert_{\mathcal{L}(\mathbb{H})} \to 0,
	\end{equation*}
	with probability 1 as $n\to\infty$.
\end{corollary}

\section{Implementation Strategy}\label{implement}
In the context of functional data analysis, we assume that the underlying sample functions, denoted as $x_{i,j}(\cdot)$, are smooth for each sample $i$ and variable $j$, where $i=1, \ldots, n$ and $j=1,\ldots, p$. However, in practical scenarios, observations are often obtained discretely on a set of grid points and contaminated by random noise. This can be represented as:
\begin{equation*}
	Y_{i,j,k} = x_{i,j}(t_k) + \varepsilon_{i,j,k}, \quad k=1,\ldots, K_j.
\end{equation*}
Here, $t_k\in\mathcal{T}_j$ and $K_j$ represents the number of discrete grid points for all samples from variable $j$. The $\varepsilon_{i,j,k}$ terms are iid random noises. To preprocess the raw data, smoothing techniques are commonly employed to convert the discrete observations into a continuous form, typically done separately for each variable and sample. One commonly used approach is the finite basis function expansion \citep[see][]{ramsay2005}.

We consider a basis system $\{\nu_j^k\}_{k\in\mathbb{N}}$ for the space $H_j$, where $j=1,\ldots,p$. By projecting each element $x_{i,j}(t)$ onto a finite-dimensional subspace $H_j^{d_j}=sp\{\nu_j^k\}_{k=1}^{d_j}\subseteq H_j$, we can represent $x_{i,j}(t)$ uniquely as:
\begin{equation*}
	x_{i,j}(t) = \pmb{\nu}_j(t)^\top\pmb{c}_{i,j}.
\end{equation*}
Here, $\pmb{c}_{i,j}=(c_{i,j,1},\ldots, c_{i,j,d_j})^\top\in\mathbb{R}^{d_j}$, $\pmb{\nu}_j=(\nu_j^1,\ldots, \nu_j^{d_j})^\top$, $t\in\mathcal{T}_j$, and $i=1,\ldots,n$. Similarly, we define $\pmb{\psi}^{(\ell)}=(\psi_1^{(\ell)},\ldots, \psi_p^{(\ell)})^\top$ with components represented as $\psi_j^{(\ell)}=\pmb{\nu}_{j}(t)^\top\pmb{b}^{(\ell)}_{j}$, where $\pmb{b}^{(\ell)}_{j}\in\mathbb{R}^{d_j}$. Using these notations, we derive the following results. First, we have:
\begin{equation}\label{e1}
	\langle x_{i,j} , \psi_j^{(\ell)} \rangle_{H_j} = \pmb{c}_{i,j}^\top \pmb{G}_j \pmb{b}^{(\ell)}_{j},
\end{equation}
where $\pmb{G}_j$ represents the \textit{Gram matrix} of the space $H_j^{d_j}$. Based on this, we obtain:
\begin{equation*}
	\langle \pmb{x}_{i} , \pmb{\psi}^{(\ell)} \rangle_{\mathbb{H}} =\sum_{j=1}^p \pmb{c}_{i,j}^\top \pmb{G}_j \pmb{b}^{(\ell)}_{j}=\pmb{c}_{i}^\top \pmb{G}\pmb{b}^{(\ell)}.
\end{equation*}
Here, $\pmb{c}_{i}^\top=(\pmb{c}_{i,1}^\top, \ldots, \pmb{c}_{i,p}^\top),$ $\pmb{b}^{(\ell)\top}=({\pmb{b}^{(\ell)\top}_{1}}, \ldots, {\pmb{b}^{(\ell)\top}_{p}})$, and $\pmb{G}=\diag\{\pmb{G}_1, \ldots , \pmb{G}_p \}.$
Furthermore, we have:
\begin{equation*}
	\langle \hat{\pmb{\mu}} , \pmb{\psi}^{(\ell)} \rangle_{\mathbb{H}} =\dfrac{1}{n}{ \bf 1}_n^\top \pmb{B}\pmb{G}\pmb{b}^{(\ell)},
\end{equation*}
where $\pmb{B}=\left[\pmb{c}_1,\ldots,\pmb{c}_n\right]^\top.$ This leads to:
\begin{equation*}
	\langle  \pmb{x}_{i}-\hat{\pmb{\mu}} , \pmb{\psi}^{(\ell)} \rangle_{\mathbb{H}} =\pmb{c}_{i}^\top \pmb{G}\pmb{b}^{(\ell)}-\dfrac{1}{n}{ \bf 1}_n^\top \pmb{B}\pmb{G}\pmb{b}^{(\ell)},
\end{equation*}
and subsequently:
\begin{align*}
	\langle {\widehat{\mathbfcal{C}}}\pmb{\psi}^{(\ell)} , \pmb{\psi}^{(\ell)} \rangle_{\mathbb{H}}&=\dfrac{1}{n-1}\sum_{i=1}^n \langle  \pmb{x}_{i}-\hat{\pmb{\mu}} , \pmb{\psi}^{(\ell)} \rangle_{\mathbb{H}}^2\\
	&=\dfrac{1}{n-1}\sum_{i=1}^n \left(\pmb{c}_{i}^\top \pmb{G}\pmb{b}^{(\ell)}-\dfrac{1}{n}{ \bf 1}_n^\top \pmb{B}\pmb{G}\pmb{
		b}^{(\ell)}\right)^2 \\
	&= {\pmb{b}^{(\ell)}}^\top\pmb{G}^\top\pmb{V}\pmb{G}\pmb{b}^{(\ell)},
\end{align*}
where $\pmb{V}=\dfrac{1}{n-1}\pmb{B}^\top \left({\bf I}_n -\dfrac{1}{n}{\bf J}_n\right)\pmb{B}$. In addition, we have:
\begin{equation}\label{e2}
	\Vert \pmb{\psi}^{(\ell)} \Vert_{\pmb \alpha}^2 = \sum_{j=1}^p \Vert \psi_j^{(\ell)} \Vert_{\alpha_j}^2 = \sum_{j=1}^p \Vert {\psi_j^{(\ell)} } \Vert_{H_j}^2 + \sum_{j=1}^p \alpha_j\Vert { \mathcal{D}^2 \psi_{j}^{(\ell)}} \Vert_{H_j}^2.
\end{equation}
Using \eqref{e1}, the first term of \eqref{e2} can be expressed as:
\begin{equation}\label{e3}
	\sum_{j=1}^p \Vert {\psi_j^{(\ell)} } \Vert_{H_j}^2 = \Vert \pmb{\psi}^{(\ell)}\Vert_{\mathbb{H}}^2 = {\pmb{b}^{(\ell)\top}}\pmb{G}\pmb{b}^{(\ell)}.
\end{equation}
Furthermore, we have:
\begin{align*}
	\mathcal{D}^2 \psi_{j}^{(\ell)}(t) &= \mathcal{D}^2 \sum_{k=1}^{d_j} b^{(\ell)}_{j,k}\nu_{j}^k(t) = \sum_{k=1}^{d_j} b^{(\ell)}_{j,k}\mathcal{D}^2\nu_{j}^k(t).
\end{align*}
From this, we obtain:
\begin{align*}
	\Vert \mathcal{D}^2\psi^{(\ell)}_{j} \Vert_{H_j}^2 &= {\pmb{b}^{(\ell)\top}_j}\pmb{D}_j\pmb{b}^{(\ell)}_j,
\end{align*}
where $\pmb{D}_j$ is the penalty matrix of the $j$-th term. Using this, the second term of \eqref{e2} becomes:
\begin{equation}\label{e4}
	\sum_{j=1}^p \alpha_j\Vert { \mathcal{D}^2 \psi_{j}^{(\ell)}} \Vert_{H_j}^2 = {\pmb{b}^{(\ell)\top}}\pmb{D}_{\pmb\alpha}\pmb{b}^{(\ell)},
\end{equation}
where $\pmb{D}_{\pmb\alpha}=\diag\{\alpha_1\pmb{D}_1, \ldots, \alpha_p\pmb{D}_p \}$. By substituting \eqref{e3} and \eqref{e4} into \eqref{e2}, we obtain:
\begin{equation*}\label{e5}
	\Vert \pmb{\psi}^{(\ell)} \Vert_{\pmb \alpha}^2 = {\pmb{b}^{(\ell)\top}}\left(\pmb{G}+\pmb{D}_{\pmb\alpha}\right)\pmb{b}^{(\ell)}.
\end{equation*}
Finally, equation \eqref{e0} can be rewritten as:
\begin{equation}\label{eq_mfpc}
	\dfrac{\langle {\widehat{\mathbfcal{C}}}{\pmb{\psi}^{(\ell)}} , \pmb{\psi}^{(\ell)} \rangle_{\mathbb{H}}}{\Vert \pmb{\psi}^{(\ell)} \Vert_{\pmb \alpha}^2} = \dfrac{{\pmb{b}^{(\ell)\top}}\pmb{G}^\top\pmb{V}\pmb{G}\pmb{b}^{(\ell)}}{{\pmb{b}^{(\ell)\top}}\left(\pmb{G}+\pmb{D}_{\pmb\alpha}\right)\pmb{b}^{(\ell)}}.
\end{equation}
The eigenequation corresponding to \eqref{eq_mfpc} can be expressed as:
\begin{align*}\label{eq_mfpc2}
	\pmb{G}^\top\pmb{V}\pmb{G}\pmb{b}^{(\ell)} = \lambda \left(\pmb{G} + \pmb{D}_{\pmb\alpha}\right)\pmb{b}^{(\ell)}.
\end{align*}
By considering a factorization $\pmb{LL}^\top=\pmb{G}+\pmb{D}_{\pmb\alpha}$ and defining $\pmb{S} = \pmb{L}^{-1}$, we obtain:
\begin{equation}\label{e6}
	\left(\pmb{S} \pmb{G}^\top\pmb{V}\pmb{G}\pmb{S}^\top\right)\left(\pmb{L}^\top\pmb{b}^{(\ell)}\right) = \lambda\pmb{L}^\top \pmb{b}^{(\ell)}.
\end{equation}
To solve \eqref{e6}, we utilize the eigenvectors $\pmb{u}^{(\ell)}$ and their corresponding eigenvalues $\lambda^{(\ell)}$ of $\pmb{S} \pmb{G}^\top\pmb{V}\pmb{G}\pmb{S}^\top$. Note that any scalar multiple of $\pmb{S}^\top\pmb{u}^{(\ell)}$ is a valid solution to \eqref{e6} with the same eigenvalue $\lambda^{(\ell)}$. However, in order to satisfy the constraint $\Vert \pmb{\psi}^{(\ell)}\Vert_{\pmb{\alpha}}^2={\pmb{b}^{(\ell)\top}}\pmb{LL}^\top\pmb{b}^{(\ell)}=1$, we normalize $\pmb{b}^{(\ell)}$ by setting $\pmb{b}^{(\ell)}=\pmb{S}^\top\pmb{u}^{(\ell)}.$

We note that the proposed method is implemented in an \proglang{R} package \pkg{ReMFPCA} (\url{https://github.com/haghbinh/ReMFPCA}).

\section{Simulation Study}\label{simul}
In this simulation study, we evaluate the performance of our ReMFPCA approach through two simulation setups. The first setup (Section \ref{fpc estimation}) focuses on assessing the estimation accuracy of functional PCs, while the second setup (Section \ref{clustering section}) evaluates the effectiveness of the proposed approach in clustering tasks.

We introduce a bivariate functional object $\mathbf{X(\pmb{t})}$, represented as $({X_1(t_1)}, $ ${X_2(t_2)})^{\top}$, and an orthonormal basis system $\pmb\psi_m \mathbf{(\pmb{t})} = \left( \psi_m^{(1)}(t_1), \psi_m^{(2)}(t_2) \right)^\top$, where $\psi_m^{(1)}(t)=\sin \left(m \pi t\right)$ and $\psi_m^{(2)}({t})=\sin \left(\frac{(2m-1)\pi}{2}t\right)$. For both simulation setups, we adopt the following functional data generating model:
\begin{equation} \label{generate}
	{\pmb{X}_i(\pmb{t})} = \sum_{m=1}^{M} {\rho}_{i,m} \pmb\psi_m \mathbf{(\pmb{t})},\qquad i=1,\ldots,n.
\end{equation}
Here, ${\rho}_{i,m}$ are iid normal random variables with mean $0$ and variance $\lambda_m$. The value of $\lambda_m$ is determined as $\lambda_m = \left(\frac{2\theta}{(2m-1)\pi}\right)^{2}$. It is worth noting that equation \eqref{generate} corresponds to the truncated multivariate Karhunen-Lo\`eve expansion of ${\pmb{X}_i(\pmb{t})}$, with the associated multivariate functional PCs being $\pmb\psi_m$ as described in \cite{happ2018multivariate}.

To simulate our observations, we consider 
\begin{equation} \label{generate_meanErr}
{\pmb{Y}_i(\pmb{t})} = \pmb{\mu}\mathbf{(\pmb{t})}+ {\pmb{X}_i(\pmb{t})} + \pmb{\epsilon}_i(\pmb{t}),
\end{equation}
where $\pmb{Y}_i(\pmb{t})$ represents the observed data at point $\pmb{t}$. The measurement errors, denoted as $\pmb{\epsilon}_i(\pmb{t})$, follow an iid normal distribution with mean $\pmb{0}$ and covariance matrix $\pmb\Sigma$:
$$\begin{bmatrix}
\sigma_1^2 & \rho\sigma_1\sigma_2  \\
\rho\sigma_1\sigma_2 & \sigma_2^2
\end{bmatrix},$$ 
where $\sigma_1^2$ and $\sigma_2^2$ represent the variances of the measurement errors for the first and second variables, respectively, and $\rho$ represents the correlation between them. By incorporating correlation in the measurement errors, we are able to capture the desired characteristics in our simulation.

\subsection{Estimation Performance} \label{fpc estimation}
In the first simulation setup, we consider $\theta = 1$ and set $\pmb{\mu}\mathbf{(\pmb{t})}$ in \eqref{generate_meanErr} to be \pmb{0}. This choice allows the first and second variables in this setup to resemble Brownian motion and Brownian bridge processes, respectively. These processes are widely recognized and extensively utilized in various fields, making them suitable and practical options for our simulation study. To generate $\pmb{\epsilon}_i$s, we set $\sigma_1 = \sigma_2 = 0.5$ and incorporate an error correlation of $\rho = 0.4$.

To assess the performance of our ReMFPCA approach, we compare it with three other methods: MFPCA, marginal FPCA and marginal ReFPCA. In the marginal approaches we separately implement the associated techniques (FPCA and ReFPCA), and then combine the results to obtain the multivariate outcome. In the conducted experiments, we generate $100$ replications, each consisting of $n=100$ observations, based on the specified setups. The accuracy of the estimated eigenvalue and eigenfunction pairs, denoted as $\hat{\lambda}_m$ and $\hat{\pmb\psi}_m$ respectively, was evaluated by comparing them to their original counterparts in each replication. This evaluation has been done using two measures: $Err(\hat{\lambda}_m) = |\hat{\lambda}_m-\lambda_m|/|\lambda_m|$ and $Err(\hat{\pmb\psi}_m) = \Vert\hat{\pmb\psi}_m-\pmb\psi\Vert_{\mathbb{H}}$.
The performance of the estimation for each method is presented in Figure \ref{sim1:plot1} and \ref{sim1:plot2} for the first eight PC ($m = 1 \dots 8$). Furthermore, the accuracy of the estimated Karhunen-Lo\'eve representation is assessed in each replication using the mean relative absolute error (MRAE), defined as $\frac{1}{n}\sum_{i=1}^{n}(\Vert\hat{\pmb{x}}_i-\pmb{x}_i\Vert_{\mathbb{H}})/\Vert\pmb{x}_i\Vert_{\mathbb{H}}$, where $\hat{\pmb{x}}_i = \sum_{m=1}^J \langle \pmb{x}_{i} , \hat{\pmb\psi}_m \rangle_{\mathbb{H}}\hat{\pmb\psi}_m.$ The resulting MRAE values are presented in the Figure \ref{sim1:plot3}, illustrating the performance of the estimation as the total number of PCs used, $J$, varies from 1 to 8. 

\begin{figure}[!t] 
  \centering
  \begin{subfigure}{0.5\textwidth}
    \centering
    \includegraphics[width=\textwidth]{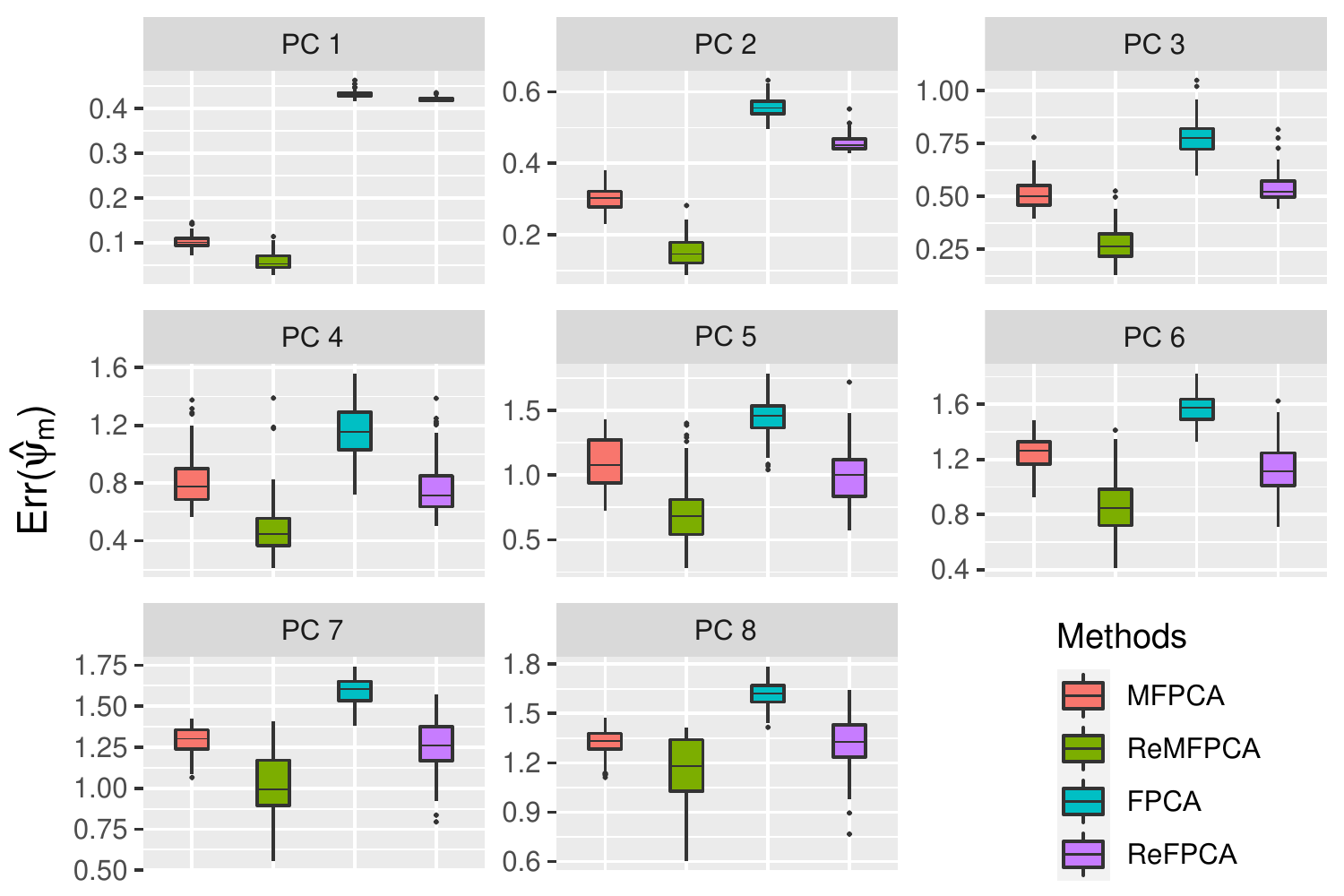}
    \caption{$Err(\hat{\pmb\psi}_m)$}
    \label{sim1:plot1}
  \end{subfigure}%
  \begin{subfigure}{0.5\textwidth}
    \centering
    \includegraphics[width=\textwidth]{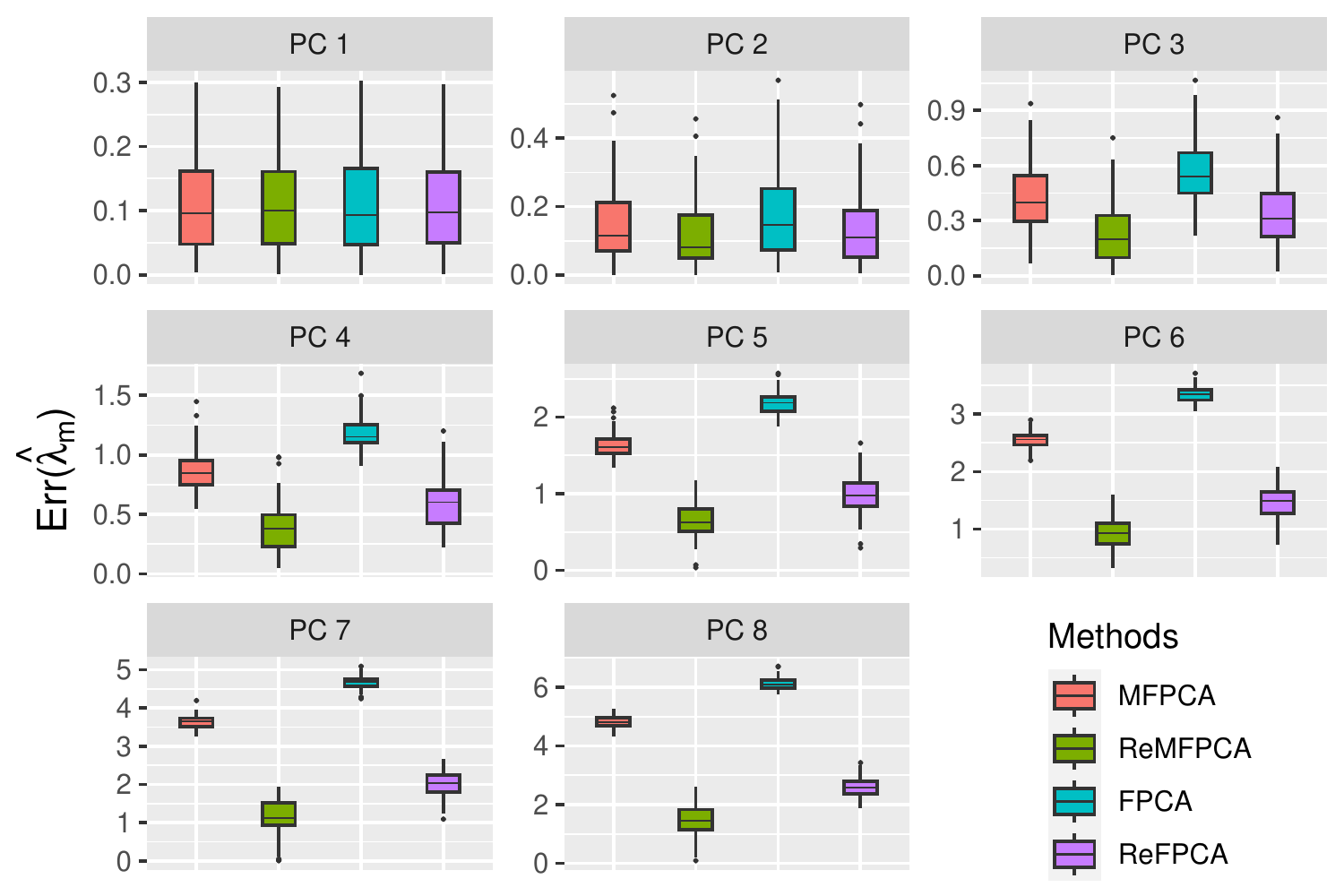}
    \caption{$Err(\hat{\lambda}_m)$}
    \label{sim1:plot2}
  \end{subfigure}
  \begin{subfigure}{0.65\textwidth}
    \centering
    \includegraphics[width=\textwidth]{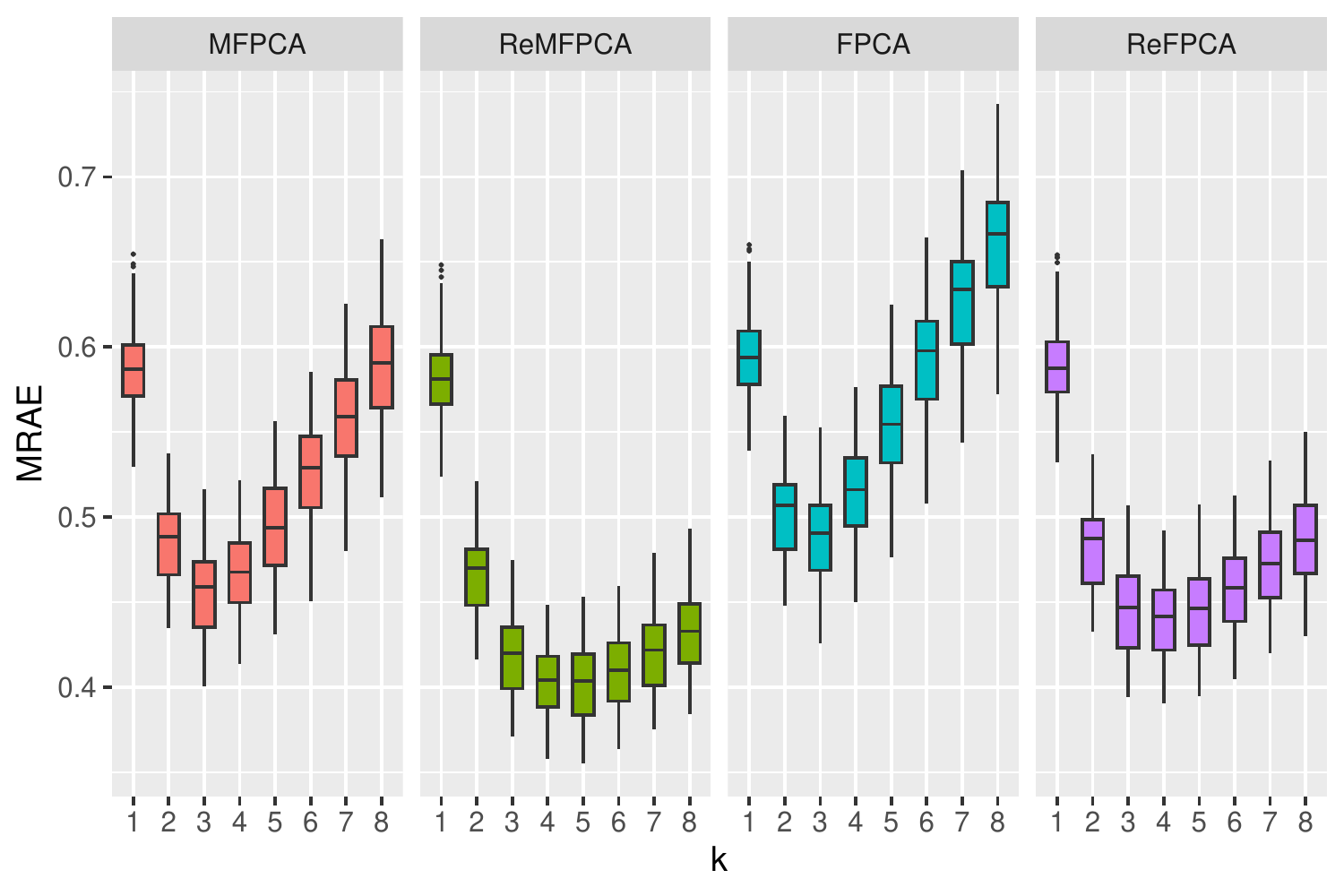}
    \caption{MRAE}
    \label{sim1:plot3}
  \end{subfigure}
  \caption{Relative errors for (a) $\hat{\pmb\psi}_m$s; (b) $\hat{\lambda}_m$s; and (c) MRAE results.}
  \label{sim1}
\end{figure}

The results depicted in Figure \ref{sim1} demonstrate the superior performance of the multivariate approaches, namely MFPCA and ReMFPCA, compared to their marginal counterparts, FPCA and ReFPCA. Furthermore, the utilization of regularization techniques, as observed in ReFPCA and ReMFPCA, leads to enhanced efficiency when compared to the non-regularized alternatives (FPCA and MFPCA). Consequently, our proposed method, ReMFPCA, exhibits a distinct advantage over the other methods.

\subsection{Clustering Performance} \label{clustering section}
In the second simulation study, we aimed to assess the suitability of the proposed approach (ReMFPCA) for clustering purposes, specifically its ability to accurately group functions with similar structures. To achieve this, we generate observations with diverse functional structures and examined how well the multivariate approaches (MFPCA and ReMFPCA) could cluster them. In the remainder of this paper, for clustering tasks, we employ k-medoids clustering on the respective PC scores, and determine the number of clusters using the Silhouette measure (see Figure \ref{silhouette}).

\begin{table}[!b] 
\centering
\resizebox{.915\textwidth}{!}{\begin{tabular}{rlllllll}
  \hline
&  &  & \multicolumn{2}{c}{ARI} & \multicolumn{2}{c}{NMI} \\ 
  \hline
M & $\theta$ & $\sigma$ & MFPCA & ReMFPCA & MFPCA & ReMFPCA \\ 
\hline
 &  & 5 & 0.9991 (0.0054) & \pmb{1} (0) & 0.9987 (0.0071) & \pmb{1} (0) \\ 
 & 1.25 & 10 & 0.9261 (0.048) & \pmb{0.9567} (0.0343) & 0.9066 (0.0575) & \pmb{0.943} (0.0439) \\ 
 &  & 15 & 0.5048 (0.1327) & \pmb{0.7413} (0.0853) & 0.4755 (0.1186) & \pmb{0.7006} (0.0883) \\ 
 \cline{2-7}
 &  & 5 & 0.9975 (0.0094) & \pmb{0.9997} (0.0032) & 0.9967 (0.0128) & \pmb{0.9996} (0.0042) \\ 
100 & 1.5 & 10 & 0.9109 (0.0526) & \pmb{0.9473} (0.037) & 0.8897 (0.0608) & \pmb{0.9317} (0.0461) \\ 
 &  & 15 & 0.4965 (0.1399) & \pmb{0.7292} (0.0855) & 0.4681 (0.1236) & \pmb{0.6892} (0.088) \\ 
  \cline{2-7}
 &  & 5 & 0.9933 (0.0141) & \pmb{0.9973} (0.0086) & 0.991 (0.0188) & \pmb{0.9963} (0.0119) \\ 
 & 2 & 10 & 0.8685 (0.0681) & \pmb{0.9198} (0.043) & 0.8424 (0.0765) & \pmb{0.8986} (0.0511) \\ 
 &  & 15 & 0.4588 (0.1373) & \pmb{0.671} (0.1095) & 0.4353 (0.1186) & \pmb{0.6333} (0.1048) \\ 
\hline
 &  & 5 & 0.9991 (0.005) & \pmb{0.9997} (0.0029) & 0.9988 (0.0071) & \pmb{0.9996} (0.0042) \\ 
 & 1.25 & 10 & 0.9336 (0.0435) & \pmb{0.962} (0.033) & 0.913 (0.0548) & \pmb{0.9494} (0.0434) \\ 
 &  & 15 & 0.5196 (0.15) & \pmb{0.7414} (0.0917) & 0.4888 (0.1353) & \pmb{0.6982} (0.0939) \\ 
  \cline{2-7}
 &  & 5 & 0.9985 (0.0064) & \pmb{0.9997} (0.0029) & 0.9979 (0.0091) & \pmb{0.9996} (0.0042) \\ 
200 & 1.5 & 10 & 0.919 (0.0478) & \pmb{0.9497} (0.039) & 0.8957 (0.0583) & \pmb{0.934} (0.0498) \\ 
 &  & 15 & 0.4998 (0.1532) & \pmb{0.726} (0.0971) & 0.4716 (0.1369) & \pmb{0.6847} (0.0984) \\ 
  \cline{2-7}
 &  & 5 & 0.9961 (0.0109) & \pmb{0.9988} (0.0059) & 0.9946 (0.0152) & \pmb{0.9983} (0.0082) \\ 
 & 2 & 10 & 0.8776 (0.0666) & \pmb{0.9213} (0.0503) & 0.8489 (0.0751) & \pmb{0.9001} (0.0608) \\ 
 &  & 15 & 0.4813 (0.1443) & \pmb{0.6881} (0.1015) & 0.4533 (0.1283) & \pmb{0.6471} (0.1006) \\ 
\hline
 &  & 5 & 0.9994 (0.0044) & \pmb{1} (0) & 0.9992 (0.0059) & \pmb{1} (0) \\ 
 & 1.25 & 10 & 0.9298 (0.0519) & \pmb{0.9568} (0.0372) & 0.9108 (0.0631) & \pmb{0.9442} (0.0466) \\ 
 &  & 15 & 0.5287 (0.1372) & \pmb{0.7439} (0.1019) & 0.4975 (0.1235) & \pmb{0.7031} (0.1062) \\ 
  \cline{2-7}
 &  & 5 & 0.9991 (0.0052) & \pmb{0.9997} (0.0029) & 0.9987 (0.0071) & \pmb{0.9996} (0.0041) \\ 
300 & 1.5 & 10 & 0.9165 (0.0527) & \pmb{0.9436} (0.0446) & 0.8953 (0.0623) & \pmb{0.928} (0.0541) \\ 
 &  & 15 & 0.5188 (0.1418) & \pmb{0.7255} (0.1065) & 0.4884 (0.1279) & \pmb{0.6846} (0.1096) \\ 
  \cline{2-7}
 &  & 5 & 0.9964 (0.0098) & \pmb{0.9988} (0.0058) & 0.995 (0.0136) & \pmb{0.9983} (0.0082) \\ 
 & 2 & 10 & 0.8694 (0.0658) & \pmb{0.9155} (0.0479) & 0.8416 (0.0732) & \pmb{0.8947} (0.0557) \\ 
 &  & 15 & 0.4873 (0.1423) & \pmb{0.6834} (0.1129) & 0.4589 (0.1277) & \pmb{0.6448} (0.1131) \\ 
   \hline
\end{tabular}}
\caption{Comparison of ReMFPCA and MFPCA performance in clustering tasks under varying values of M, $\theta$, and $\sigma$. Mean values are reported, with standard errors shown in parentheses.}
\label{cluster_table}
\end{table}

To create varied functional structures, we manipulated the parameter $\pmb{\mu}\mathbf{(\pmb{t})}$ in equation \eqref{generate_meanErr}. In our simulation, we introduced a matrix $\pmb{A}$ defined as follows:
\begin{equation*}
\pmb{A} = \begin{bmatrix}
4 & 0 & 4 \\
4 & 4 & 0 \\
0 & 4 & 4 \\
\end{bmatrix}.
\end{equation*}
By multiplying $[\pmb{\psi}_m\mathbf{(\pmb{t})}]_{m=1}^3$ with $\pmb{A}$, we generated three distinct mean functions denoted as $[\pmb{\mu}_i\mathbf{(\pmb{t})}]_{i=1}^3$. Each class ($\pmb{\mu}_i\mathbf{(\pmb{t})}$) had the same size ($n/3$) in each replication. For the simulations, we assumed a correlation of $\rho=0.4$ and set $\sigma_1=\sigma_2=\sigma$ (referred to as $\sigma$). By adjusting the parameters $\theta$, $\sigma$, and $M$, we were able to generate observations with varying degrees of smoothness, where higher values of these parameters resulted in rougher observations. 

To comprehensively compare the performance of ReMFPCA and MFPCA, we evaluate their outcomes using two widely accepted external measures: the adjusted rand index (ARI) \citep{hubert1985comparing} and the normalized mutual information (NMI) \citep{kvalseth1987entropy}. The summary statistics of ReMFPCA and MFPCA performance, based on these measures, are presented in Table \ref{cluster_table}. Our analysis involves 100 replications, each comprising 99 observations, across various scenarios characterized by different values of M, $\theta$, and $\sigma$. Across all experiments, ReMFPCA consistently outperforms MFPCA, yielding higher ARI and NMI values. Notably, when confronted with observations exhibiting high roughness, ReMFPCA demonstrates superior performance compared to the MFPCA approach.

\begin{figure}[!b] 
\centering
  \includegraphics[width=.59\textwidth]{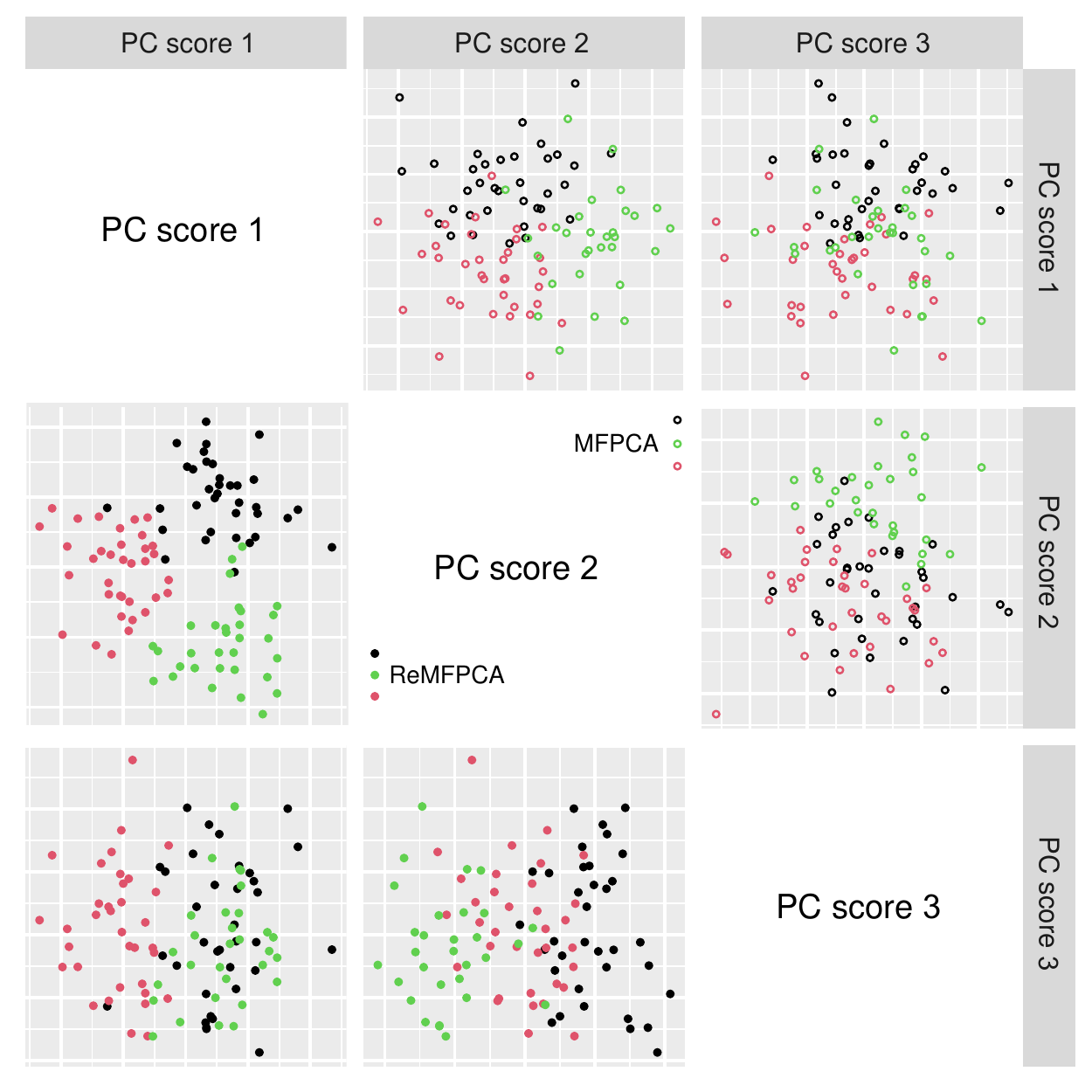}
  \caption{Pairwise scatterplot of the first three PC scores for a randomly selected repetition. Lower triangular portion: results from ReMFPCA approach; upper triangular portion: results from MFPCA approach.}
  \label{clustering_compare}
\end{figure}

To provide visual insights into the results, Figure \ref{clustering_compare} displays the first three principal component (PC) scores obtained from ReMFPCA and MFPCA. The illustration corresponds to a randomly selected replication with $M = 300$, $\theta = 1$, and $\sigma = 15$. It is evident that the ReMFPCA method exhibits distinct segregation into three groups, indicating its potential effectiveness for clustering purposes.

\section{Application to power consumption dataset}\label{application}
In this section, we consider a dataset contains measurements of electric power consumption in one household located in Sceaux (7km of Paris, France) with a 5-minute sampling rate between December 2006 and November 2010 \citep{HouseholdElectricPower}. This dataset includes two variables, active power and voltage, with a total of $n = 1440$ observations. To ensure both variables contribute equal amounts of variation to the analysis, we follow \cite{happ2018multivariate} to rescale variables via weight $w_j$:
$$w_j = \left(\int_{\mathcal{T}_i} \widehat{Var}(X_j(t))dt\right)^{-1},\qquad j = 1,2.$$
Using these weights, the integrated variance equals 1 for the rescaled variable $\tilde{X}_j(t) = w_j^{1/2}X_j(t).$

\begin{figure}[!bh]
    \centering
    \includegraphics[width=.98\textwidth]{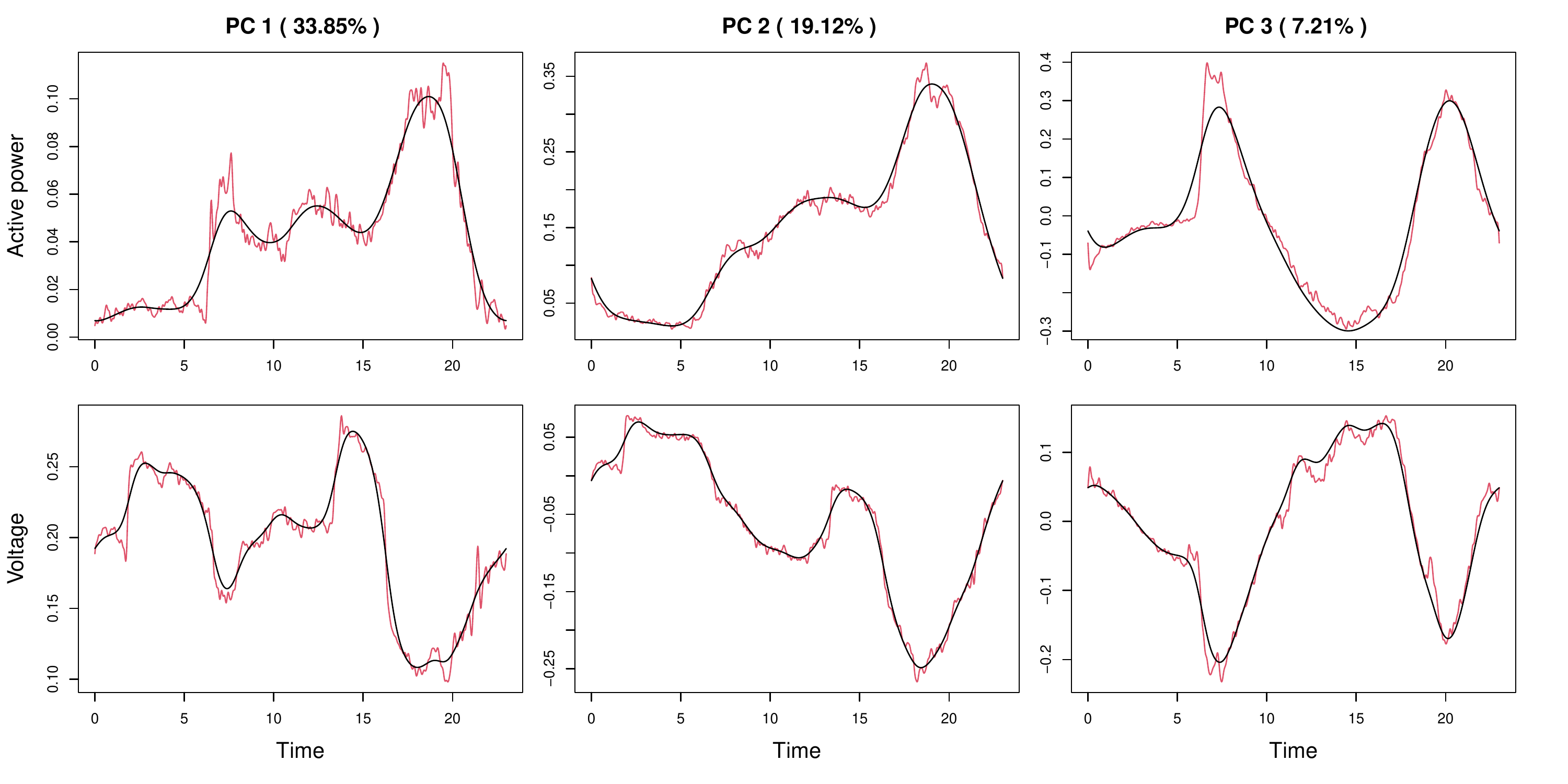}
    \caption{First three functional PCs of MFPCA (red) and ReMFPCA (black).}
    \label{FPC plot:plot1}
  \end{figure}

\begin{figure}[!t]
  \begin{subfigure}[b]{0.5\linewidth}
    \centering
    \includegraphics[width=\textwidth]{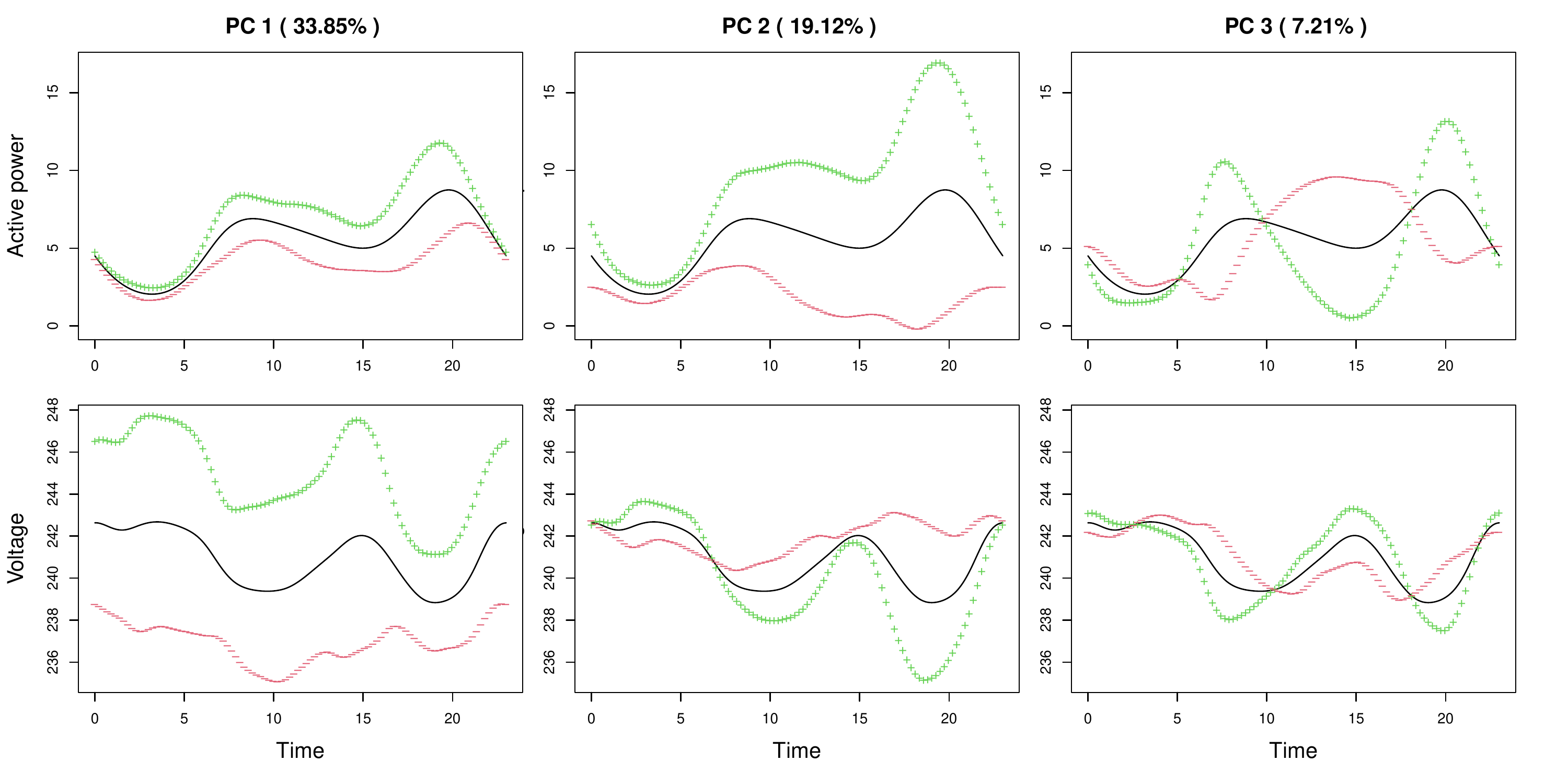}
    \caption{ReMPCA}
    \label{FPC plot:plot2}
  \end{subfigure}
    \begin{subfigure}[b]{0.5\linewidth}
    \centering
    \includegraphics[width=\textwidth]{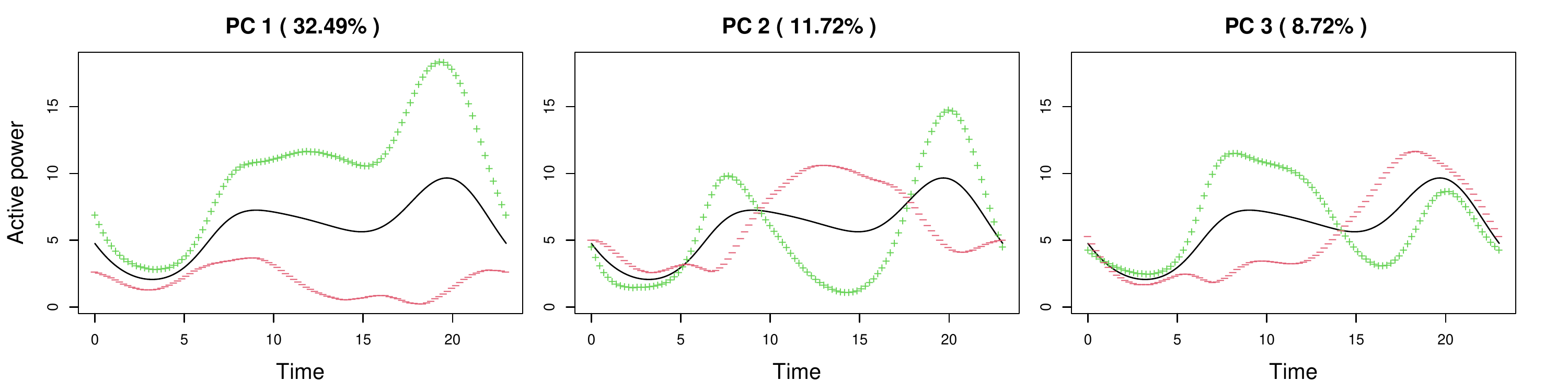}
    \includegraphics[width=\textwidth]{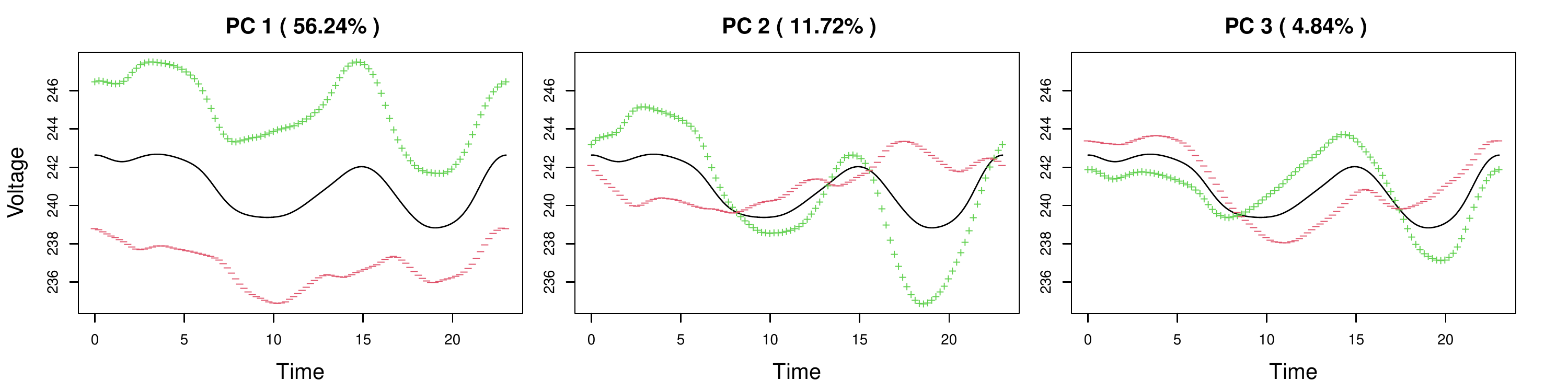}
    \caption{ReFPCA}
    \label{FPC plot:plot3}
  \end{subfigure}
  \caption{Effect on the overall mean curve of adding and subtracting a suitable multiple of each of the first three PCs from (a) ReMFPCA, and (b) ReFPCA.}
  \label{FPC plot}
\end{figure}
 
\begin{figure}[!b] 
  \centering
  \begin{subfigure}{0.5\textwidth}
    \centering
    \includegraphics[width=\textwidth]{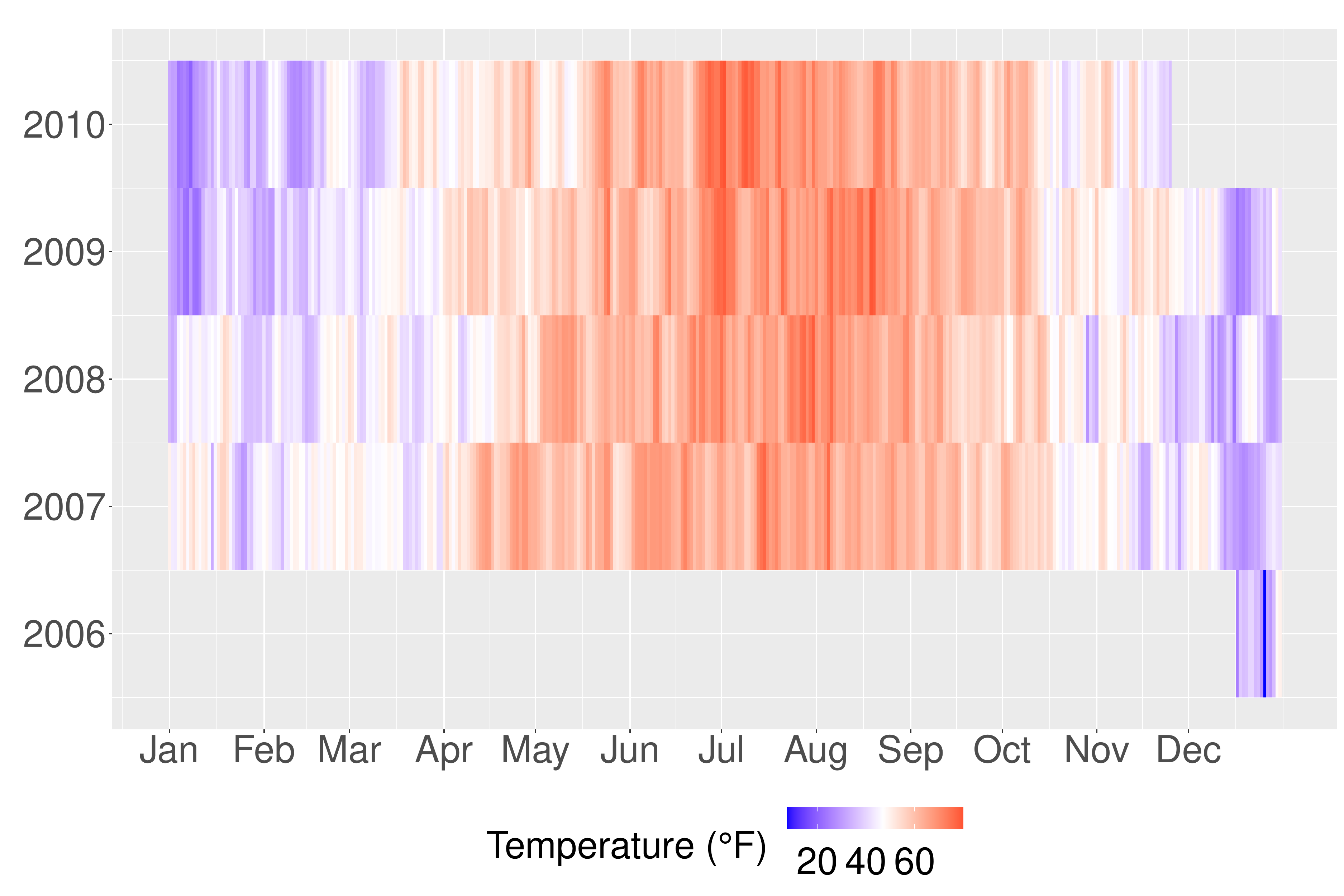}
    \caption{Temperature}
    \label{FPC1 heatmap:plot1}
  \end{subfigure}%
  \begin{subfigure}{0.5\textwidth}
    \centering
    \includegraphics[width=\textwidth]{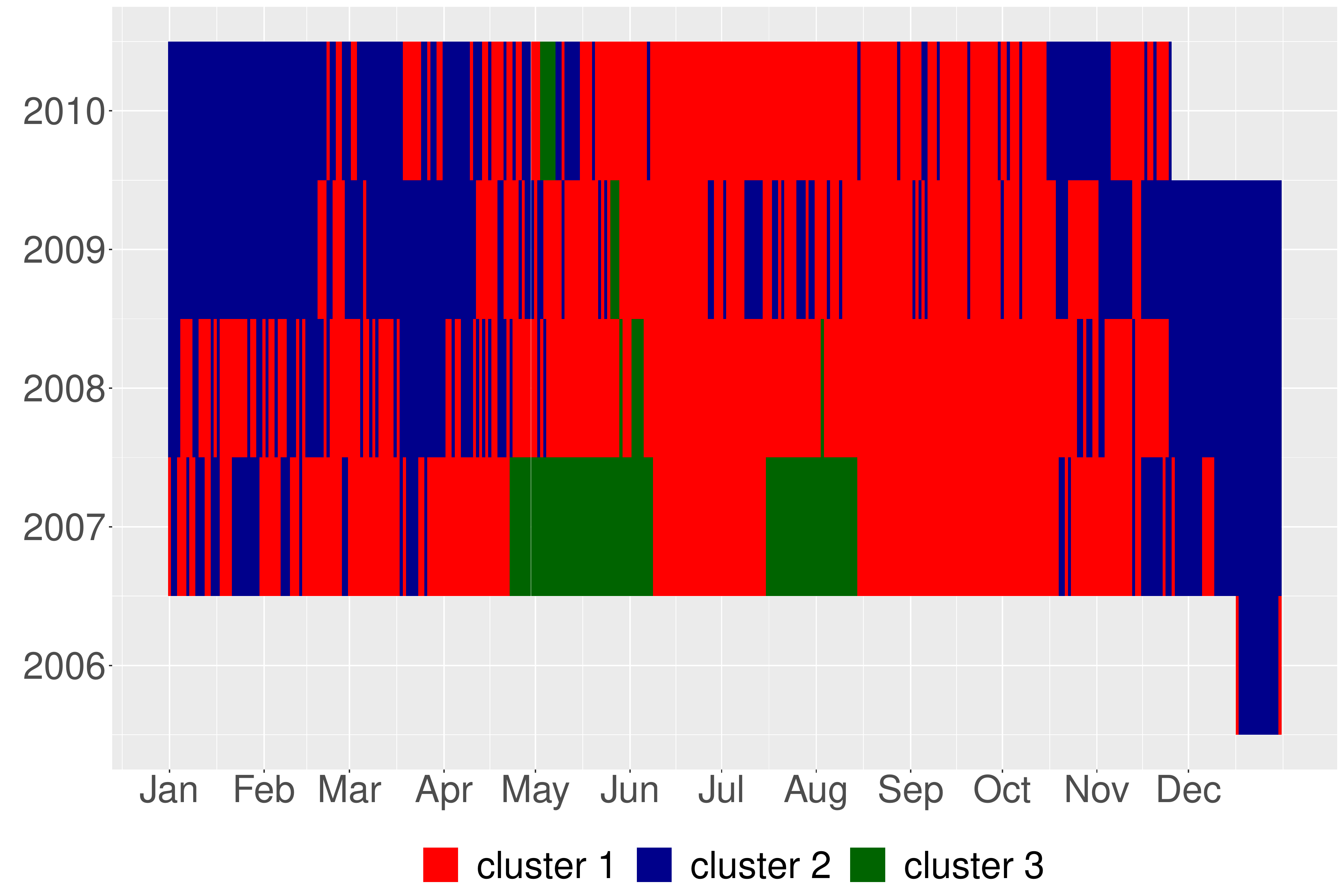}
    \caption{ReMFPCA}
    \label{FPC1 heatmap:plot2}
  \end{subfigure}
    \begin{subfigure}{0.5\textwidth}
    \centering
    \includegraphics[width=\textwidth]{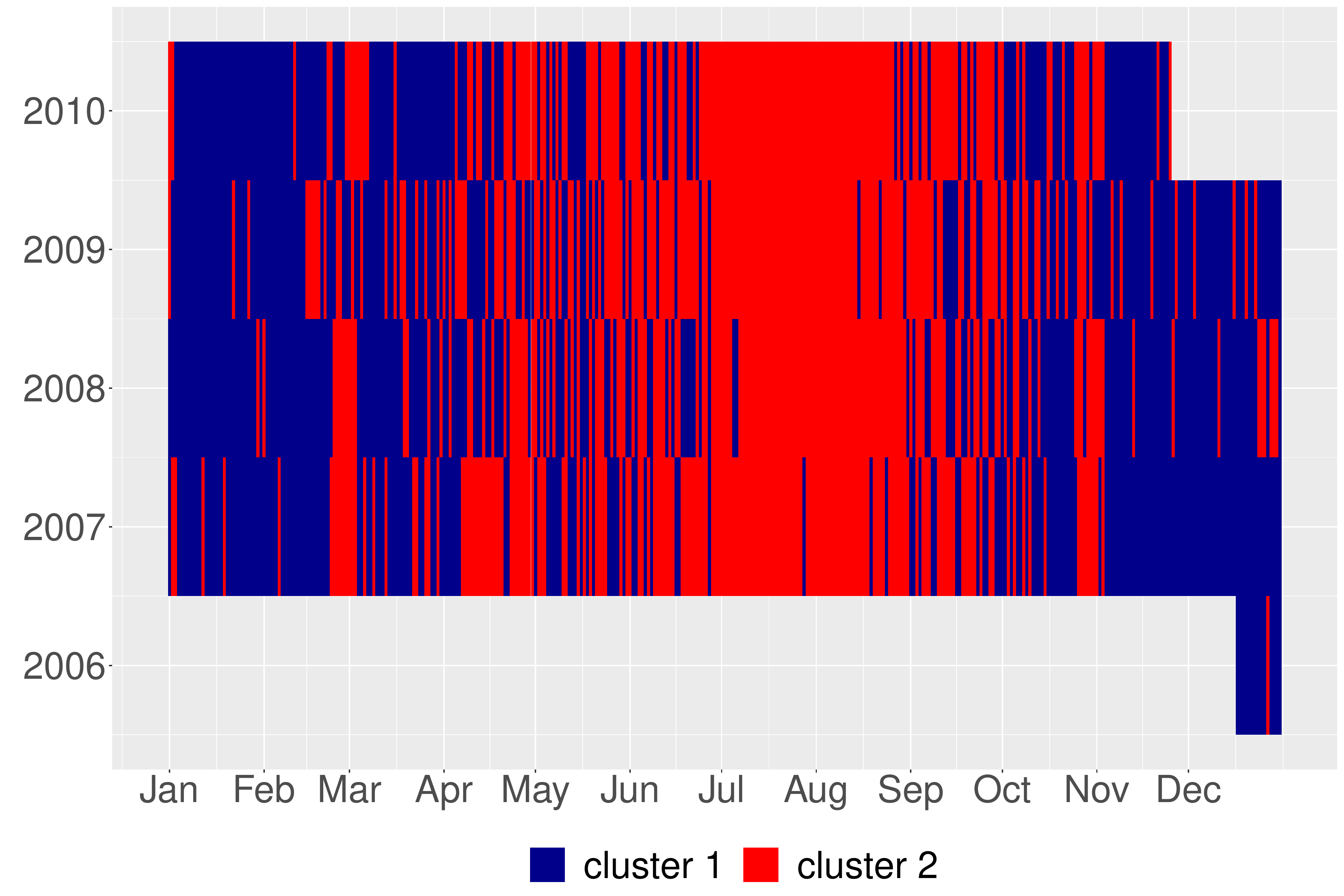}
    \caption{ReFPCA of Active Power}
    \label{FPC1 heatmap:plot3}
  \end{subfigure}%
  \begin{subfigure}{0.5\textwidth}
    \centering
    \includegraphics[width=\textwidth]{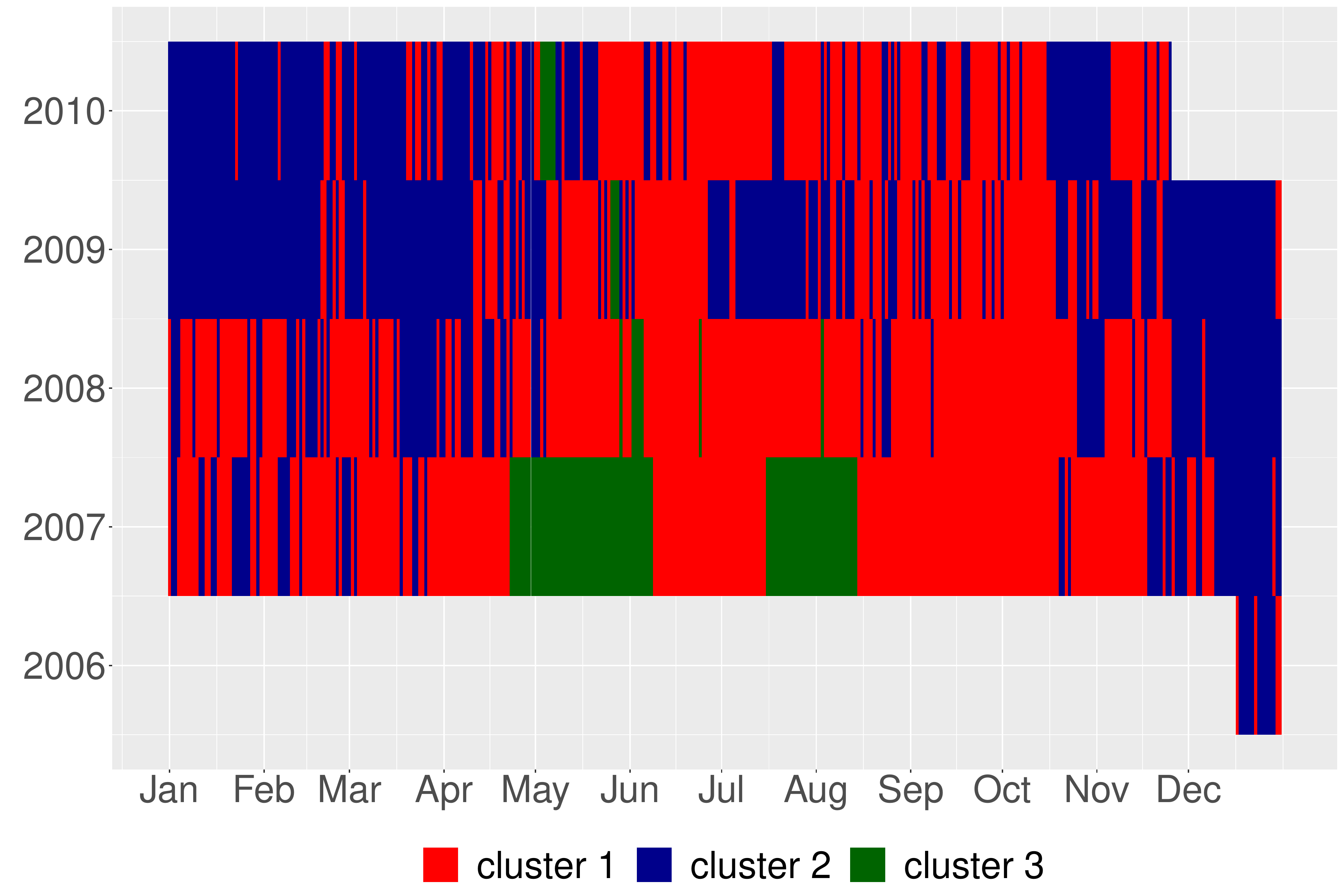}
    \caption{ReFPCA of Voltage}
    \label{FPC1 heatmap:plot4}
  \end{subfigure}
  \caption{(a) Average temperature heatmap; Clustering results based on first PC for (b) ReMFPCA; (c) ReFPCA of Active Power; and (d) ReFPCA of Voltage.}
  \label{FPC1 heatmap}
\end{figure}

The comparison of the first three principal components (PCs) obtained from ReMFPCA and MFPCA is shown in Figure \ref{FPC plot:plot1}. Notably, the ReMFPCA method effectively addresses the issue of significant roughness in the extracted PCs. We further compare the results of ReMFPCA with marginal ReFPCA in Figure \ref{FPC plot}, which illustrates the impact of adding and subtracting a multiple of each smoothed PC on the mean curve for both methods. It can be observed that for the first PCs of ReMFPCA and ReFPCA, there is an approximate constant effect, indicating that the shape of the mean functional curve remains largely unchanged. The overall pattern of the PCs associated with the voltage variable exhibits strong similarity between ReMFPCA and ReFPCA (second row of Figure \ref{FPC plot}). Regarding the active power variable, the 2nd and 3rd PCs of ReMFPCA closely resemble the 1st and 2nd PCs of ReFPCA, respectively. However, the 1st PC associated with the active power in ReMFPCA appears to differ from those in ReFPCA. Consequently, it may be worthwhile to explore the impact of this difference in a clustering task.

Figure \ref{FPC1 heatmap} displays the heatmaps depicting the average temperature of Sceaux throughout the study period \citep{wunderground}, along with the clustering results obtained from the first PC scores of three methods: ReMFCA, ReFPCA of Active Power, and ReFPCA of Voltage. The heatmaps reveal that the first PC of ReMFPCA effectively captures the annual temperature pattern, as evident from the presence of distinct blue and red clusters. Additionally, it successfully identifies a specific scenario characterized by low voltage, represented by the green cluster (for more details, see Figure \ref{FPC1: cluster}). In contrast, the first PC of ReFPCA fails to capture certain relevant information, likely due to its inability to incorporate the dependence between variables. For instance, the first PC of ReFPCA for active power fails to detect the low voltage scenario, while the first PC of ReFPCA for voltage exhibits limited performance in identifying the summer season of 2009.

\begin{figure}[!b] 
\centering
  \includegraphics[width=\textwidth]{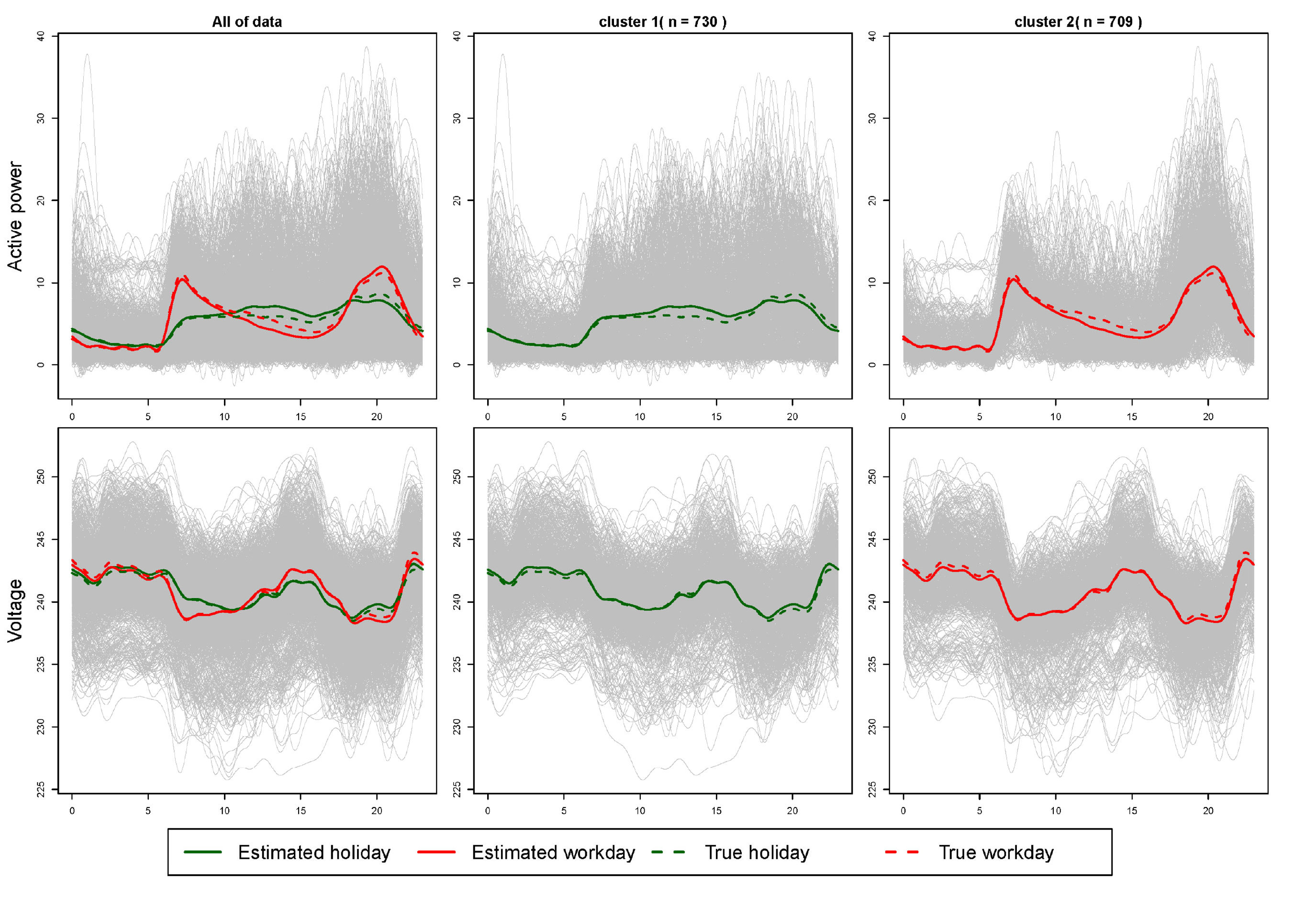}
  \caption{Clustering results based on 3rd PC of ReMFPCA with estimated cluster mean curve (solid line) and true cluster mean curve (dashed line).}
  \label{FPC3: cluster}
\end{figure}

The clustering results based on the 2nd PC score for ReMFPCA exhibit a high degree of similarity to the results obtained using the first PC score (see Figure \ref{FPC2 cluster}). In the case of the 3rd PC of ReMFPCA, a notable distinction emerges between working hours (10am - 6pm) and the remaining time periods. This pattern, particularly prominent during working hours, allows for the identification of workday and holiday behavior. The interpretability of this finding is supported by Figure \ref{FPC3: cluster}, which showcases the clustering results corresponding to the 3rd PC scores of ReMFPCA. The plot provides compelling evidence, revealing a remarkable resemblance between the estimated cluster mean curves and the true mean curve associated with holidays (including school holidays) or workdays. Furthermore, similar patterns are observed in the 2nd PC of ReFPCA for active power and the 3rd PC of ReFPCA for voltage. However, Figure \ref{FPC3 boxplot} provides additional evidence that ReMFPCA surpasses ReFPCA in effectively distinguishing between workdays and holidays.

\begin{figure}[!t] 
  \centering
    \includegraphics[width=.31\textwidth]{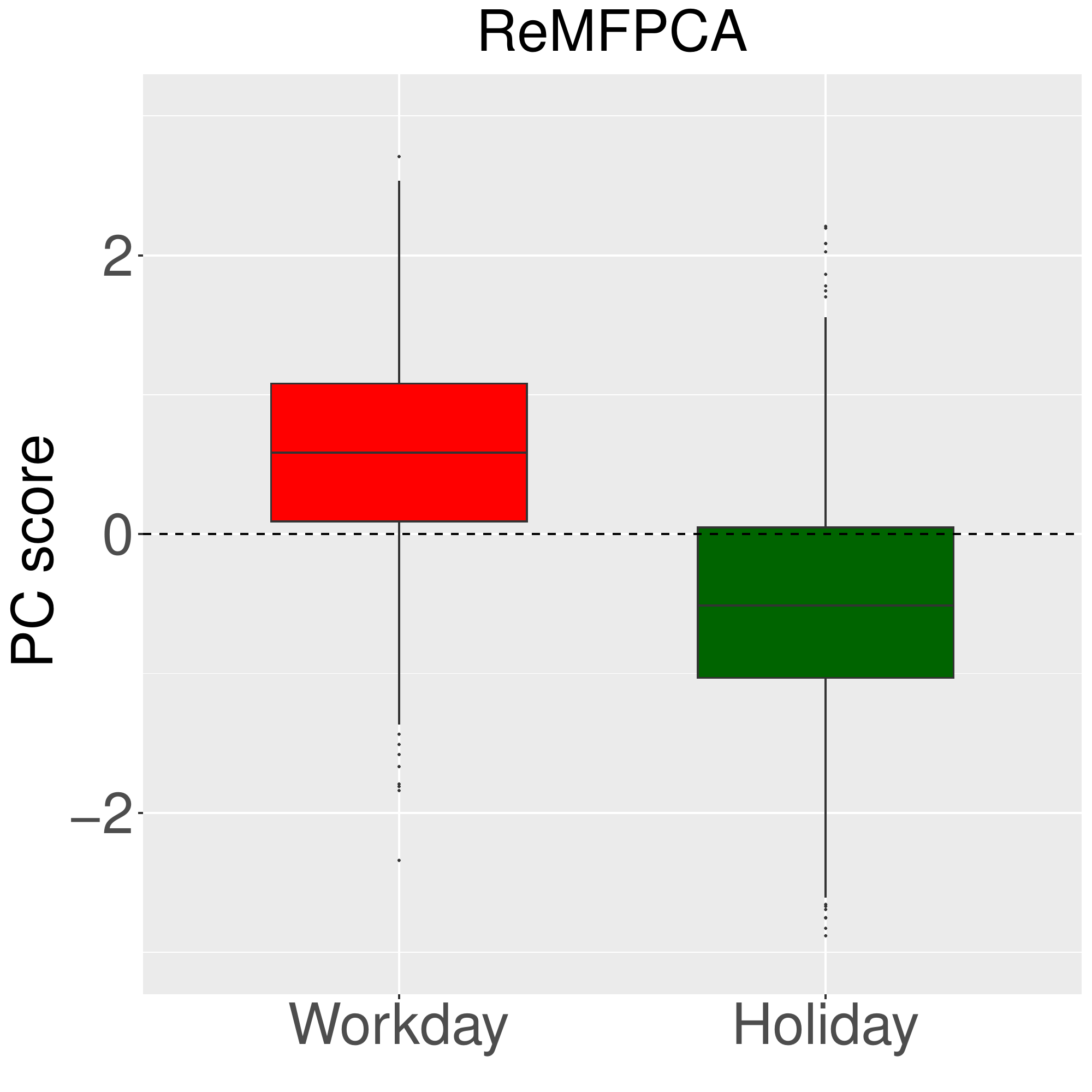}
    \includegraphics[width=.31\textwidth]{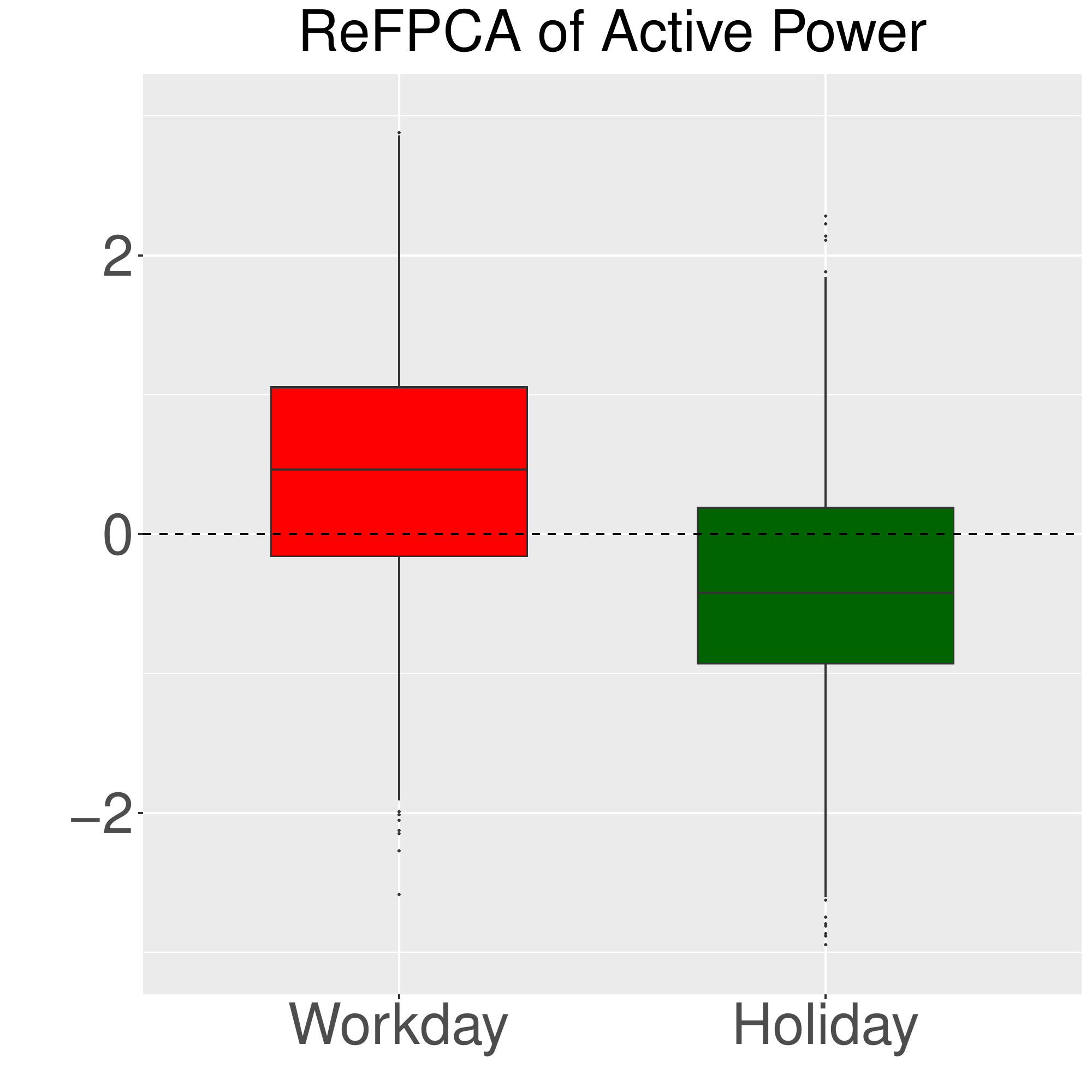}
    \includegraphics[width=.31\textwidth]{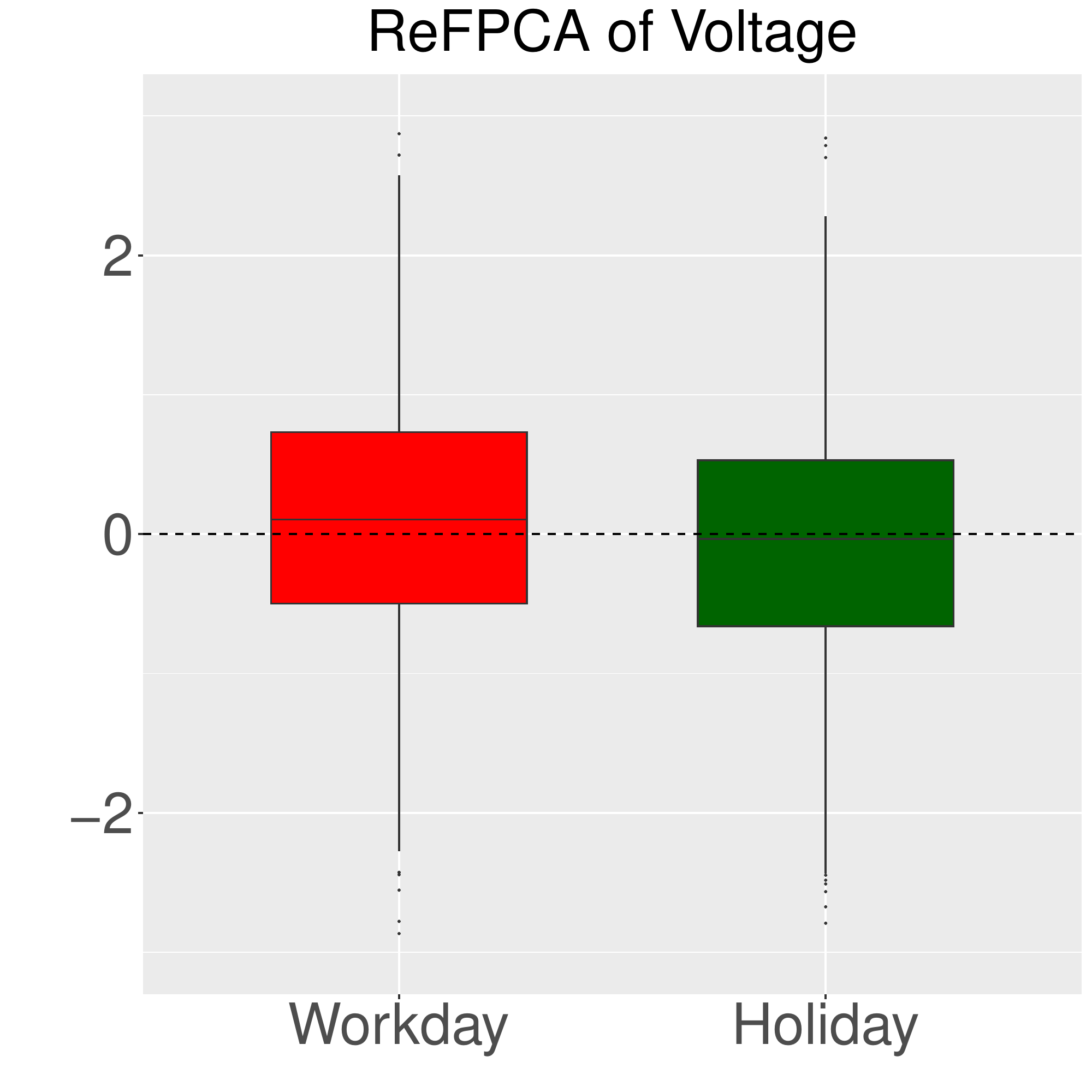}
  \caption{Boxplots of the corresponding PC scores for ReMFPCA and ReFPCA.}
  \label{FPC3 boxplot}
\end{figure}

\section{Conclusion}\label{conclusion}
In summary, this paper introduced ReMFPCA to address the challenge of defining an appropriate smoothness level for multivariate functional data. By incorporating smoothness constraints as penalty terms within a regularized optimization framework, ReMFPCA generates smoothed multivariate functional PCs that offer improved interpretability and stability. Through simulations and real data examples, we demonstrated the effectiveness of ReMFPCA compared to other methods such as MFPCA, marginal FPCA, and regularized FPCA (ReFPCA). ReMFPCA showcased superior performance in terms of interpretability, stability, and accuracy of results. Its ability to capture shared patterns of variation provides valuable insights into complex multivariate functional datasets across various domains. 

A limitation of the proposed technique lies in the computational cost associated with determining the tuning parameter vector ${\pmb \alpha}$, which has a length of $p$. In this study, cross-validation was employed to select the tuning parameter, which proved to be computationally expensive within the current framework. Future research efforts may focus on devising more efficient approaches for obtaining the smoothing parameters in order to address this issue.

Overall, ReMFPCA contributes to the growing field of functional data analysis, empowering researchers to gain deeper insights from complex multivariate functional datasets in domains such as biomedical imaging, environmental science, and finance.

{\center \section*{Appendix}}
\appendix 
\renewcommand{\thesection}{\Alph{section}}
\setcounter{equation}{0}
\renewcommand\theequation{\Alph{section}.\arabic{equation}}
\setcounter{figure}{0}
\renewcommand\thefigure{\Alph{section}.\arabic{figure}}
\section{Proofs}
	\begin{proof}[Proof of Theorem \ref{theorem-existence}]
	The sample covariance operator is a bounded operator due to its finite rank. Let $\mathbb{W}^2$ be the Hilbert space equipped with the inner product $\langle . ,.\rangle_{\pmb{\alpha}}$. For any $\pmb{x, y}\in\mathbb{W}^2$, the inner product $\langle {\mathbfcal C}\pmb{x} , \pmb{y}\rangle_{\mathbb{H}}$ defines a bilinear form in $\mathbb{W}^2$ and can be bounded as follows:
\begin{equation*}
\left| \langle {\mathbfcal C}\pmb{x} , \pmb{y}\rangle_{\mathbb{H}}\right|
\leq 
\Vert {\mathbfcal C} \Vert_{\mathcal{L}(\mathbb{H})} 
\Vert \pmb{x} \Vert_{\mathbb{H}}
\Vert \pmb{y} \Vert_{\mathbb{H}}
\leq
\Vert {\mathbfcal C} \Vert_{\mathcal{L}(\mathbb{H})} 
\Vert \pmb{x} \Vert_{\pmb{\alpha}}
\Vert \pmb{y} \Vert_{\pmb{\alpha}}.
\end{equation*}
Based on Theorem 5.35 of \cite{weidmann1980linear}, there exists a unique bounded operator ${\mathbfcal C}_{\pmb{\alpha}}\in\mathcal{L}({\mathbb{W}^2})$ such that:
\begin{equation*}
\langle {\mathbfcal C}\pmb{x} , \pmb{y}\rangle_{\mathbb{H}}
=
\langle {\mathbfcal C}_{\pmb{\alpha}}\pmb{x} , \pmb{y}\rangle_{\pmb{\alpha}}.
\end{equation*}
Moreover, ${\mathbfcal C}_{\pmb{\alpha}}$ is symmetric since:
\begin{equation*}
\langle {\mathbfcal C}_{\pmb{\alpha}}\pmb{x} , \pmb{y}\rangle_{\pmb{\alpha}}
=
\langle {\mathbfcal C}\pmb{x} , \pmb{y}\rangle_{\mathbb{H}}
=
\langle \pmb{x} ,{\mathbfcal C} \pmb{y}\rangle_{\mathbb{H}}
=
\langle {\mathbfcal C} \pmb{y},\pmb{x}\rangle_{\mathbb{H}}
=
\langle {\mathbfcal C}_{\pmb{\alpha}}\pmb{y} , \pmb{x}\rangle_{\pmb{\alpha}}
=
\langle  \pmb{x},{\mathbfcal C}_{\pmb{\alpha}}\pmb{y}\rangle_{\pmb{\alpha}}.
\end{equation*}
Furthermore, ${\mathbfcal C}_{\pmb{\alpha}}$ is positive definite. To demonstrate that ${\mathbfcal C}_{\pmb{\alpha}}$ is a compact operator, we need to show that for any bounded sequence $\lbrace \pmb x_m\in{\mathbb{W}^2,\ m\in\mathbb{N}} \rbrace$, there exists a subsequence $\lbrace \pmb x_{mk} \rbrace$ such that:
\begin{equation}\label{compactness}
\langle {\mathbfcal C}_{\pmb{\alpha}}\pmb x_{mk}-\pmb x_{ml} , \pmb x_{mk}-\pmb x_{ml}\rangle_{\pmb{\alpha}}
=
\langle {\mathbfcal C}\pmb x_{mk}-\pmb x_{ml}, \pmb x_{mk}-\pmb x_{ml}\rangle_{\mathbb{H}} \to 0,
\end{equation}
as $k,\ l\to\infty$ \citep{riesz1990functional}. Since ${\mathbfcal C}$ is a compact operator and $\lbrace \pmb x_{m} \rbrace$ is a bounded sequence in $L^2_\mathbb{H}$, we can select a subsequence $\lbrace \pmb x_{mk} \rbrace$ such that it converges and \eqref{compactness} holds. Thus, ${\mathbfcal C}_{\pmb{\alpha}}$ is a compact operator, and according to the Hilbert-Schmidt Theorem, it possesses eigenvalues and eigenfunctions $\lbrace ({\lambda}_i^{[\pmb\alpha]},  {\pmb\psi}_i^{[\pmb\alpha]}): i\in\mathbb{N} \rbrace$, which are solutions of the successive optimization problems \eqref{e0} and \eqref{e0.5}.
\end{proof}

\begin{proof}[Proof of Lemma \ref{half-smoothing}]
Since we have,
\begin{equation*}
	\mathbfcal{S}_{\pmb{\alpha}}^{-2}=
	\begin{bmatrix}
		\mathcal{I}+{\alpha_1}\mathcal{Q}&0&\ldots&0\\
		0&\mathcal{I}+{\alpha_2}\mathcal{Q}&\ddots&\vdots\\
		\vdots&\ddots&\ddots&0\\
		0&\ldots&0&\mathcal{I}+{\alpha_p}\mathcal{Q}
	\end{bmatrix}
	=\mathbfcal{I}+\diag\left({\alpha_1}\mathcal{Q},\ldots, {\alpha_p}\mathcal{Q}\right),
\end{equation*}
then
\begin{align*}
	\langle\mathbfcal{S}_{\pmb{\alpha}}^{-1}\pmb{x}, \mathbfcal{S}_{\pmb{\alpha}}^{-1}\pmb{y}\rangle_{\mathbb{H}}^2&=
	\langle \pmb{x},\mathbfcal{S}_{\pmb{\alpha}}^{-2}\pmb{y} \rangle_{\mathbb{H}}
	=\langle \pmb{x},\pmb{y} \rangle_{\mathbb{H}}+\sum_{i=1}^p \alpha_i\langle x_i,y_i''''\rangle_{H_i} =
	\langle{\pmb{x ,y}}\rangle_{\pmb{\alpha}}^2.
\end{align*}	
\end{proof}

\begin{proof}[Proof of Lemma \ref{rmfpca_obtain}]
Using the result of the Lemma \ref{half-smoothing} we get
\begin{equation*}
	\dfrac{\langle {{\mathbfcal{C}}}\pmb{\varphi} , \pmb{\varphi} \rangle_{\mathbb{H}}}{\Vert \pmb{\varphi} \Vert_{\pmb \alpha}^2}=
	\dfrac{\langle {{\mathbfcal{C}}}\mathbfcal{S}_{\pmb{\alpha}}\mathbfcal{S}_{\pmb{\alpha}}^{-1}\pmb{\varphi} , \mathbfcal{S}_{\pmb{\alpha}}\mathbfcal{S}_{\pmb{\alpha}}^{-1}\pmb{\varphi} \rangle_{\mathbb{H}}}{\Vert \mathbfcal{S}_{\pmb{\alpha}}^{-1}\pmb{\varphi} \Vert_{\mathbb{H}}^2}
	=
	\dfrac{\langle\left(\mathbfcal{S}_{\pmb{\alpha}} {{\mathbfcal{C}}}\mathbfcal{S}_{\pmb{\alpha}}\right)\mathbfcal{S}_{\pmb{\alpha}}^{-1}\pmb{\varphi} , \mathbfcal{S}_{\pmb{\alpha}}^{-1}\pmb{\varphi} \rangle_{\mathbb{H}}}{\Vert \mathbfcal{S}_{\pmb{\alpha}}^{-1}\pmb{\varphi} \Vert_{\mathbb{H}}^2}.
\end{equation*}

To complete the proof, we utilize the results of \eqref{mfpca1} and \eqref{mfpca2}. 
\end{proof}

\begin{proof}[Proof of Theorem \ref{theorem_cons}.]
The proof of Theorem \ref{theorem_cons} follows an induction approach, similar to \cite{silverman1996smoothed}. We define the statement $\mathfrak{A}_k$ for any positive integer $k$, which comprises the following three convergences:
\begin{align}
	\hat{\lambda}_k^{[\pmb\alpha]} &\to {\lambda}_k, \label{conv4}\\
	\sum_{j=1}^{p}\alpha_j\Vert\mathcal{D}^2\hat{\psi}_{kj}^{[\pmb\alpha]}\Vert_{H_j}^2&\to 0, \ \text{and} \label{conv5}\\
	\langle \tilde{\pmb\psi}_k^{[\pmb\alpha]} , {\pmb\psi}_k \rangle_{\mathbb{H}} &\to 1. \label{conv6}
\end{align}
Here, $\hat{\pmb\psi}_{k}^{[\pmb\alpha]}=(\hat{\psi}_{k1}^{[\pmb\alpha]}, \ldots, \hat{\psi}_{kp}^{[\pmb\alpha]})$. It is worth noting that since $\Vert  \hat{\pmb\psi}_{k}^{[\pmb\alpha]}\Vert_{\pmb{\alpha}}^2=1$, the limit \eqref{conv5} is equivalent to:

\begin{equation}\label{conv7}
	\Vert  \hat{\pmb\psi}_{k}^{[\pmb\alpha]}\Vert_{\mathbb{H}}^2\to 1.
\end{equation}
Furthermore, the limit \eqref{conv6} can be equivalently expressed as:
\begin{equation}\label{conv8}
	\Vert {\mathbfcal P}_{{\pmb\psi}_k} -{\mathbfcal P}_{\hat{\pmb\psi}_k^{[\pmb\alpha]} } \Vert_{\mathcal{L}(\mathbb{H})} \to 0.
\end{equation}
Let us consider a fixed $j\ge1$ and assume, by induction, that $\mathfrak{A}_k$ holds for all $k<j$ (For $j=1$, there is no assumption to be made). Now, we introduce the projection ${\mathbfcal P}$, which is perpendicular to ${\pmb\psi}_1, \ldots, {\pmb\psi}_{j-1}$:
\begin{equation*}
	{\mathbfcal P}\pmb{x}= \pmb{x}-\sum_{i=1}^{j-1}\langle{\pmb\psi}_{i},\pmb{x} \rangle_{\mathbb{H}}{\pmb\psi}_{i},
\end{equation*}
and the projection ${\mathbfcal P}^{[\pmb\alpha]}$, which is perpendicular to $\hat{\pmb\psi}^{[\pmb\alpha]}_1, \ldots, \hat{\pmb\psi}_{j-1}^{[\pmb\alpha]}$ with respect to the inner product $\langle.,. \rangle_{\pmb\alpha}$:
\begin{equation*}
	{\mathbfcal P}^{[\pmb\alpha]}\pmb{x}= \pmb{x}-\sum_{i=1}^{j-1}\langle\hat{\pmb\psi}_{i}^{[\pmb\alpha]},\pmb{x} \rangle_{\pmb\alpha}\hat{\pmb\psi}_{i}^{[\pmb\alpha]}.
\end{equation*}
The following lemma is essential to finalize the proof of Theorem \ref{theorem_cons}.

\begin{lemma}\label{lem_sup}
For any $j\geq 1$, assuming $\mathfrak{A}k$ holds for all $k<j$, let $\mathbfcal P$ and $\mathbfcal P^{[\pmb\alpha]}$ be defined as above. Then, we have
\begin{equation}\label{conv9}
\sup_{\pmb x: \Vert \pmb x\Vert_{\pmb\alpha}\leq 1}\left|\langle {\mathbfcal P}\pmb{x},\mathbfcal{C}{\mathbfcal P}\pmb{x}\rangle_{\mathbb{H}}-\langle {\mathbfcal P}^{[\pmb\alpha]}\pmb{x},\widehat{\mathbfcal{C}}{\mathbfcal P}^{[\pmb\alpha]}\pmb{x}\rangle_{\mathbb{H}}\right|\to 0.
\end{equation}
\end{lemma}
\begin{proof}[Proof of Lemma \ref{lem_sup}.]
We consider $j\geq 2$. Throughout the proof, we focus on $\pmb{x} \in \mathbb{H}$. For each $k<j$, applying the triangle and Cauchy-Schwarz inequalities, and utilizing the fact that $\Vert \pmb x\Vert_{\mathbb{H}} \leq \Vert \pmb x\Vert_{\pmb\alpha}$, we obtain\\	
\resizebox{\textwidth}{!}{$\begin{aligned}
		\sup&_{\pmb x: \Vert \pmb x\Vert_{\pmb\alpha}\leq 1}\Vert\langle{\pmb\psi}_{k},\pmb{x} \rangle_{\mathbb{H}}{\pmb\psi}_{k}-\langle\hat{\pmb\psi}_{k}^{[\pmb\alpha]},\pmb{x} \rangle_{\pmb\alpha}\hat{\pmb\psi}_{k}^{[\pmb\alpha]}\Vert_{\mathbb{H}}\\
		&\leq
		\sup_{\pmb x: \Vert \pmb x\Vert_{\pmb\alpha}\leq 1}\Vert\langle{\pmb\psi}_{k},\pmb{x} \rangle_{\mathbb{H}}{\pmb\psi}_{k}-\langle\hat{\pmb\psi}_{k}^{[\pmb\alpha]},\pmb{x} \rangle_{\mathbb{H}}\hat{\pmb\psi}_{k}^{[\pmb\alpha]}\Vert_{\mathbb{H}}
		+\Vert  \hat{\pmb\psi}_{k}^{[\pmb\alpha]}\Vert_{\mathbb{H}}\sup_{\pmb x: \Vert \pmb x\Vert_{\pmb\alpha}\leq 1}
		\sum_{j=1}^{p}\alpha_j
		\langle\mathcal{D}^2\hat{\psi}_{kj}^{[\pmb\alpha]},\mathcal{D}^2{x_j} \rangle_{H_j}\\
		&\leq
		\sup_{\pmb x: \Vert \pmb x\Vert_{\mathbb{H}}\leq 1}\Vert\langle{\pmb\psi}_{k},\pmb{x} \rangle_{\mathbb{H}}{\pmb\psi}_{k}-\langle\tilde{\pmb\psi}_{k}^{[\pmb\alpha]},\pmb{x} \rangle_{\mathbb{H}}\tilde{\pmb\psi}_{k}^{[\pmb\alpha]}\Vert_{\mathbb{H}}+
		\sup_{\pmb x: \Vert \pmb x\Vert_{\mathbb{H}}\leq 1}\Vert\langle\tilde{\pmb\psi}_{k}^{[\pmb\alpha]},\pmb{x} \rangle_{\mathbb{H}}\tilde{\pmb\psi}_{k}^{[\pmb\alpha]}-\langle\hat{\pmb\psi}_{k}^{[\pmb\alpha]},\pmb{x} \rangle_{\mathbb{H}}\hat{\pmb\psi}_{k}^{[\pmb\alpha]}\Vert _{\mathbb{H}}\\		
		&+\Vert  \hat{\pmb\psi}_{k}^{[\pmb\alpha]}\Vert_{\mathbb{H}}\sup_{\pmb x: \Vert \pmb x\Vert_{\pmb\alpha}\leq 1}
		\sum_{j=1}^{p}\alpha_j
		\Vert\mathcal{D}^2\hat{\psi}_{kj}^{[\pmb\alpha]}\Vert_{H_j}\Vert\mathcal{D}^2{x_j} \Vert_{H_j}\\
		&=
		\sup_{\pmb x: \Vert \pmb x\Vert_{\mathbb{H}}\leq 1}\Vert {\mathbfcal P}_{{\pmb\psi}_k}\pmb x -{\mathbfcal P}_{\hat{\pmb\psi}_k^{[\pmb\alpha]} }\pmb x \Vert_{\mathbb{H}}+
		\sup_{\pmb x: \Vert \pmb x\Vert_{\mathbb{H}}\leq 1}\Vert\langle\tilde{\pmb\psi}_{k}^{[\pmb\alpha]},\pmb{x} \rangle_{\mathbb{H}}\tilde{\pmb\psi}_{k}^{[\pmb\alpha]}-\Vert  \hat{\pmb\psi}_{k}^{[\pmb\alpha]}\Vert_{\mathbb{H}}^2\langle\tilde{\pmb\psi}_{k}^{[\pmb\alpha]},\pmb{x} \rangle_{\mathbb{H}}\tilde{\pmb\psi}_{k}^{[\pmb\alpha]}\Vert _{\mathbb{H}}\\		
		&+\Vert  \hat{\pmb\psi}_{k}^{[\pmb\alpha]}\Vert_{\mathbb{H}}\sup_{\pmb x: \Vert \pmb x\Vert_{\pmb\alpha}\leq 1}
		\sum_{j=1}^{p}\alpha_j
		\Vert\mathcal{D}^2\hat{\psi}_{kj}^{[\pmb\alpha]}\Vert_{H_j}\Vert\mathcal{D}^2{x_j} \Vert_{H_j}\\
		&=
		\Vert {\mathbfcal P}_{{\pmb\psi}_k} -{\mathbfcal P}_{\hat{\pmb\psi}_k^{[\pmb\alpha]} } \Vert_{\mathcal{L}(\mathbb{H})}+
		\Vert {\mathbfcal P}_{\hat{\pmb\psi}_k^{[\pmb\alpha]} } \Vert_{\mathcal{L}(\mathbb{H})}
		\left(1-\Vert  \hat{\pmb\psi}_{k}^{[\pmb\alpha]}\Vert_{\mathbb{H}}^2\right)\qquad\qquad\quad\ \ \ \ \\	
		&\qquad+\Vert  \hat{\pmb\psi}_{k}^{[\pmb\alpha]}\Vert_{\mathbb{H}}\sup_{\pmb x: \Vert \pmb x\Vert_{\pmb\alpha}\leq 1}
		\sum_{j=1}^{p}\alpha_j
		\Vert\mathcal{D}^2\hat{\psi}_{kj}^{[\pmb\alpha]}\Vert_{H_j}\Vert\mathcal{D}^2{x_j} \Vert_{H_j}.
	\end{aligned}$}	
In the last term since for each $i$, $\Vert x_i\Vert_{H_i} \leq \Vert{\pmb{x}}\Vert_{\mathbb{H}} \leq \Vert \pmb{x}\Vert_{\pmb{\alpha}},$ using \eqref{conv5} we get
	\begin{equation*}
		\sup_{\pmb x: \Vert \pmb x\Vert_{\pmb\alpha}\leq 1}
		\sum_{j=1}^{p}\alpha_j
		\Vert\mathcal{D}^2\hat{\psi}_{kj}^{[\pmb\alpha]}\Vert_{H_j}\Vert\mathcal{D}^2{x_j} \Vert_{H_j} \leq \sum_{j=1}^{p}\alpha_j
		\Vert\mathcal{D}^2\hat{\psi}_{kj}^{[\pmb\alpha]}\Vert_{H_j}\to 0.
	\end{equation*}
	Using this and Equations \eqref{conv7} and \eqref{conv8} we obtain
	\begin{align} \label{conv95}
		\sup_{\pmb x: \Vert \pmb x\Vert_{\pmb\alpha}\leq 1}&\Vert\langle{\pmb\psi}_{k},\pmb{x} \rangle_{\mathbb{H}}{\pmb\psi}_{k}-\langle\hat{\pmb\psi}_{k}^{[\pmb\alpha]},\pmb{x} \rangle_{\pmb\alpha}\hat{\pmb\psi}_{k}^{[\pmb\alpha]}\Vert_{\mathbb{H}}\to 0\ \text{for all } k<j. 
	\end{align}
	By summing over $k<j$ we have
	\begin{equation}\label{conv10}
		\sup_{\pmb x}\dfrac{	\Vert ({\mathbfcal P}^{[\pmb\alpha]} - {\mathbfcal P})\pmb x\Vert_{\mathbb{H}}}{\Vert \pmb x\Vert_{\pmb\alpha}}
		\leq
		\sum_{k=1}^{j-1}\sup_{\pmb x: \Vert \pmb x\Vert_{\pmb\alpha}\leq 1}\Vert\langle{\pmb\psi}_{k},\pmb{x} \rangle_{\mathbb{H}}{\pmb\psi}_{k}-\langle\hat{\pmb\psi}_{k}^{[\pmb\alpha]},\pmb{x} \rangle_{\pmb\alpha}\hat{\pmb\psi}_{k}^{[\pmb\alpha]}\Vert_{\mathbb{H}}
		\to 0.
	\end{equation}
In a seminal work, \cite{dauxois1982asymptotic} demonstrated the convergence of $\Vert \hat{\mathbfcal{C}} -\mathbfcal{C}\Vert_{\mathcal{L}(\mathbb{H})}$ by utilizing the strong law of large numbers in Hilbert space. Building upon this finding and leveraging the uniform convergence result given by \eqref{conv10}, we establish \eqref{conv9} and successfully conclude the proof of the lemma.	
\end{proof}
\noindent We can now prove the three parts of $\mathfrak{A}_j$ successively.\\ \vspace{-.03in}

\noindent\textit{Proof of \eqref{conv4} for $k=j$}. Please note that the maximum value of $\langle {\mathbfcal P}\pmb{x},\mathbfcal{C}{\mathbfcal P}\pmb{x}\rangle_{\mathbb{H}}$ over the set $\Vert \pmb x\Vert_{\mathbb{H}}\leq 1$ is $\lambda_j$, and it is achieved at ${\pmb\psi}{j}$. Similarly, the maximum value of $\langle {\mathbfcal P}^{[\pmb\alpha]}\pmb{x},\widehat{\mathbfcal{C}}{\mathbfcal P}^{[\pmb\alpha]}\pmb{x}\rangle_{\mathbb{H}}$ over the set $\Vert \pmb x\Vert_{\pmb\alpha}\leq 1$ is $\lambda_j^{[\pmb\alpha]}$, attained at $\hat{\pmb\psi}_{j}^{[\pmb\alpha]}$. Using property \eqref{conv9} we have
\begin{align} \label{conv105}
	\lambda_j&=\langle {\mathbfcal P}{\pmb\psi}_{j},\mathbfcal{C}{\mathbfcal P}{\pmb\psi}_{j}\rangle_{\mathbb{H}}
	\ge
	\langle {\mathbfcal P}\tilde{\pmb\psi}_j^{[\pmb\alpha]} ,\mathbfcal{C}{\mathbfcal P}\tilde{\pmb\psi}_j^{[\pmb\alpha]} \rangle_{\mathbb{H}}
	=
	\frac{\langle {\mathbfcal P}\hat{\pmb\psi}_j^{[\pmb\alpha]} ,\mathbfcal{C}{\mathbfcal P}\hat{\pmb\psi}_j^{[\pmb\alpha]} \rangle_{\mathbb{H}}}{\Vert  \hat{\pmb\psi}_{k}^{[\pmb\alpha]}\Vert_{\mathbb{H}}^2}\\ \notag
	&\ge
	\langle {\mathbfcal P}\hat{\pmb\psi}_j^{[\pmb\alpha]} ,\mathbfcal{C}{\mathbfcal P}\hat{\pmb\psi}_j^{[\pmb\alpha]} \rangle_{\mathbb{H}}=
	\langle {\mathbfcal P}^{[\pmb\alpha]}\hat{\pmb\psi}_j^{[\pmb\alpha]} ,\hat{\mathbfcal{C}}{\mathbfcal P}^{[\pmb\alpha]}\hat{\pmb\psi}_j^{[\pmb\alpha]} \rangle_{\mathbb{H}}+o(1)=\lambda_j^{[\pmb\alpha]}+o(1).
\end{align}
Also, we have
\begin{align*}
	\lambda_j^{[\pmb\alpha]}&=\langle {\mathbfcal P}^{[\pmb\alpha]}\hat{\pmb\psi}_j^{[\pmb\alpha]} ,\hat{\mathbfcal{C}}{\mathbfcal P}^{[\pmb\alpha]}\hat{\pmb\psi}_j^{[\pmb\alpha]} \rangle_{\mathbb{H}}
	\ge
	\langle {\mathbfcal P}^{[\pmb\alpha]}{{\pmb\psi}_j}/{\Vert  {\pmb\psi}_{j}\Vert_{\pmb\alpha}} ,\hat{\mathbfcal{C}}{\mathbfcal P}^{[\pmb\alpha]}{{\pmb\psi}_j }/{\Vert  {\pmb\psi}_{j}\Vert_{\pmb\alpha}}\rangle_{\mathbb{H}}\\
	&=
	\langle {\mathbfcal P}{{\pmb\psi}_j}/{\Vert  {\pmb\psi}_{j}\Vert_{\pmb\alpha}} ,{\mathbfcal{C}}{\mathbfcal P}{{\pmb\psi}_j }/{\Vert  {\pmb\psi}_{j}\Vert_{\pmb\alpha}}\rangle_{\mathbb{H}}+o(1)={\lambda_j}/{\Vert  {\pmb\psi}_{j}\Vert_{\pmb\alpha}^2}+o(1)=\lambda_j+o(1),
\end{align*}
and the fact that $\Vert  {\pmb\psi}_{j}\Vert_{\pmb\alpha}\to 1$, since $\alpha_i\to 0.$ The proof of \eqref{conv4} is completed for $k=j.$\\

\noindent \textit{Proof of \eqref{conv5} for $k=j$}.
It immediately follows from
\begin{equation*}
	\sum_{j=1}^{p}\alpha_j\Vert\mathcal{D}^2\hat{\psi}_{kj}^{[\pmb\alpha]}\Vert_{H_j}^2=1- \Vert\hat{\pmb\psi}_{j}^{[\pmb\alpha]}\Vert_{\mathbb{H}}^2=
	1-\dfrac{\langle {\mathbfcal P}\hat{\pmb\psi}_j^{[\pmb\alpha]} ,\mathbfcal{C}{\mathbfcal P}\hat{\pmb\psi}_j^{[\pmb\alpha]} \rangle_{\mathbb{H}}}{\langle {\mathbfcal P}\tilde{\pmb\psi}_j^{[\pmb\alpha]} ,\mathbfcal{C}{\mathbfcal P}\tilde{\pmb\psi}_j^{[\pmb\alpha]} \rangle_{\mathbb{H}}}\to 0.
\end{equation*}\\

\noindent \textit{Proof of \eqref{conv6} for $k=j$}.
We begin by considering the case when $j\geq 2$ and substitute $x=\hat{\pmb\psi}_{j}^{[\pmb\alpha]}$ into equation \eqref{conv95}. Given that $\Vert  \hat{\pmb\psi}_{k}^{[\pmb\alpha]}\Vert_{\pmb{\alpha}}^2=1$, and for each $k<j$,
\begin{align*}
	\Vert\langle{\pmb\psi}_{k},\hat{\pmb\psi}_{j}^{[\pmb\alpha]} \rangle_{\mathbb{H}}{\pmb\psi}_{k}-\langle\hat{\pmb\psi}_{k}^{[\pmb\alpha]},\hat{\pmb\psi}_{j}^{[\pmb\alpha]} \rangle_{\pmb\alpha}\hat{\pmb\psi}_{k}^{[\pmb\alpha]}\Vert_{\mathbb{H}}
	=
	\Vert\langle{\pmb\psi}_{k},\hat{\pmb\psi}_{j}^{[\pmb\alpha]} \rangle_{\mathbb{H}}{\pmb\psi}_{k}\Vert_{\mathbb{H}}
	=
	|\langle{\pmb\psi}_{k},\hat{\pmb\psi}_{j}^{[\pmb\alpha]} \rangle_{\mathbb{H}}|,
\end{align*}
we can deduce from this and equation \eqref{conv7} for $k=j$ that, for each $k<j$,
\begin{align*}
	\lim_{n\to\infty}\langle{\pmb\psi}_{k},\tilde{\pmb\psi}_{j}^{[\pmb\alpha]} \rangle_{\mathbb{H}}
	=
	\lim_{n\to\infty}\langle{\pmb\psi}_{k},\hat{\pmb\psi}_{j}^{[\pmb\alpha]} \rangle_{\mathbb{H}}
	=0,
\end{align*}
implying that
\begin{align}\label{conv107}
	\sum_{k<j}\langle{\pmb\psi}_{k},\tilde{\pmb\psi}_{j}^{[\pmb\alpha]} \rangle_{\mathbb{H}}\to 0.
\end{align}
Next, we expand $\tilde{\pmb\psi}_{j}^{[\pmb\alpha]}$ in terms of the complete orthonormal sequence ${\pmb\psi}_{j}$. We have
\begin{align}\label{conv11}
	\mathbfcal {P}\tilde{\pmb\psi}_{j}^{[\pmb\alpha]}=	\mathbfcal {P}\sum_{i=1}^\infty\langle{\tilde{\pmb\psi}_{j}^{[\pmb\alpha]}},\pmb\psi_{i} \rangle_{\mathbb{H}}\pmb\psi_{i}
	&=
	\sum_{i=1}^\infty\langle{\tilde{\pmb\psi}_{j}^{[\pmb\alpha]}},\pmb\psi_{i} \rangle_{\mathbb{H}}\mathbfcal {P}\pmb\psi_{i}
	=
	\sum_{i=j}^\infty\langle{\tilde{\pmb\psi}_{i}^{[\pmb\alpha]}},\pmb\psi_{i} \rangle_{\mathbb{H}}\pmb\psi_{i},
\end{align}
where we utilize the fact that $\mathbfcal {P}\pmb\psi_{i}=0$ for $i<j$ and $\mathbfcal {P}\pmb\psi_{i}=1$ for $i\geq j$. By employing the property $\mathbfcal C\pmb\psi_{i}=\lambda_i\pmb\psi_{i}$, we have
\begin{align}\label{conv12}
	\mathbfcal C\mathbfcal {P}\tilde{\pmb\psi}_{j}^{[\pmb\alpha]}=	\sum_{i=j}^\infty\langle{\tilde{\pmb\psi}_{i}^{[\pmb\alpha]}},\pmb\psi_{i} \rangle_{\mathbb{H}}\mathbfcal C\pmb\psi_{i}
	=
	\sum_{i=j}^\infty\lambda_i\langle{\tilde{\pmb\psi}_{i}^{[\pmb\alpha]}},\pmb\psi_{i} \rangle_{\mathbb{H}}\pmb\psi_{i}.
\end{align}
Combining equations \eqref{conv11} and \eqref{conv12} and utilizing the orthonormality of $\pmb\psi_{i}$, we obtain the following expression:
\begin{align}\label{conv13}
\langle\mathbfcal {P}\tilde{\pmb\psi}{j}^{[\pmb\alpha]},\mathbfcal C\mathbfcal {P}\tilde{\pmb\psi}{j}^{[\pmb\alpha]}\rangle_\mathbb{H}=\sum_{i=j}^\infty\lambda_i\langle{\tilde{\pmb\psi}{j}^{[\pmb\alpha]}},\pmb\psi{i} \rangle_{\mathbb{H}}^2.
\end{align}
Considering that $\mathbfcal {P}$ is a projection, we have $\Vert\mathbfcal {P}\tilde{\pmb\psi}{j}^{[\pmb\alpha]}\Vert_{\mathbb{H}}^2\leq\Vert\tilde{\pmb\psi}{j}^{[\pmb\alpha]}\Vert_{\mathbb{H}}^2=1$. Substituting this into equation \eqref{conv13}, we find:
\begin{align}\label{conv14}
	\lambda_j&-\langle\mathbfcal {P}\tilde{\pmb\psi}_{j}^{[\pmb\alpha]},\mathbfcal C\mathbfcal {P}\tilde{\pmb\psi}_{j}^{[\pmb\alpha]}\rangle_\mathbb{H}
	\ge
	\lambda_j\Vert\mathbfcal {P}\tilde{\pmb\psi}_{j}^{[\pmb\alpha]}\Vert_{\mathbb{H}}^2
	-
	\sum_{i=j}^\infty\lambda_i\langle{\tilde{\pmb\psi}_{j}^{[\pmb\alpha]}},\pmb\psi_{i} \rangle_{\mathbb{H}}^2
	\notag\\&=
	\lambda_j\sum_{i=j}^\infty\langle{\tilde{\pmb\psi}_{j}^{[\pmb\alpha]}},\pmb\psi_{i} \rangle_{\mathbb{H}}^2
	-
	\sum_{i=j}^\infty\lambda_i\langle{\tilde{\pmb\psi}_{j}^{[\pmb\alpha]}},\pmb\psi_{i} \rangle_{\mathbb{H}}^2
	=
	\sum_{i=j}^\infty(\lambda_j-\lambda_i)\langle{\tilde{\pmb\psi}_{j}^{[\pmb\alpha]}},\pmb\psi_{i} \rangle_{\mathbb{H}}^2
	\\&\ge
	\sum_{i=j}^\infty(\lambda_j-\lambda_{j+1})\langle{\tilde{\pmb\psi}_{j}^{[\pmb\alpha]}},\pmb\psi_{i} \rangle_{\mathbb{H}}^2\ge 0,\notag
\end{align}
since $\lambda_j$ is a decreasing sequence. As the inequalities in \eqref{conv105} tend to equalities, we have $\lambda_j-\langle\mathbfcal {P}\tilde{\pmb\psi}{j}^{[\pmb\alpha]},\mathbfcal C\mathbfcal {P}\tilde{\pmb\psi}{j}^{[\pmb\alpha]}\rangle_\mathbb{H}\to 0$. Consequently, all the inequalities in \eqref{conv14} tend to equalities. Since $\lambda_j\neq\lambda_{j+1}$, it follows that:
\begin{equation}\label{conv15}
\sum_{i>j}^\infty\langle{\tilde{\pmb\psi}{j}^{[\pmb\alpha]}},\pmb\psi{i} \rangle_{\mathbb{H}}^2 \to 0.
\end{equation}
Combining equations \eqref{conv107} and \eqref{conv15} with the property $\sum_{i=1}^\infty\langle{\tilde{\pmb\psi}{j}^{[\pmb\alpha]}},\pmb\psi{i} \rangle_{\mathbb{H}}^2=1$, we can conclude that $\langle{\tilde{\pmb\psi}{j}^{[\pmb\alpha]}},\pmb\psi{i} \rangle_{\mathbb{H}}^2\to 1$. This completes the proof of equation \eqref{conv6} for $k=j$. By following the inductive argument outlined above, we have now completed the proof of Theorem \ref{theorem_cons}.
\end{proof}

\begin{proof}[Proof of Corollary \ref{corollary_31}]
	Using the Theorem \ref{theorem_cons} we have
\begin{align*}
	\Vert \tilde{\pmb\psi}_i^{[\pmb\alpha]} - {\pmb\psi}_i \Vert_{\mathbb{H}}^2&=\Vert \tilde{\pmb\psi}_i^{[\pmb\alpha]}\Vert_{\mathbb{H}}^2+\Vert  {\pmb\psi}_i \Vert_{\mathbb{H}}^2-2\langle \tilde{\pmb\psi}_i^{[\pmb\alpha]} , {\pmb\psi}_i \rangle_{\mathbb{H}}=2-2\langle \tilde{\pmb\psi}_i^{[\pmb\alpha]} , {\pmb\psi}_i \rangle_{\mathbb{H}}\to 0,
\end{align*}
\end{proof}

\begin{proof}[Proof of Corollary \ref{corollary_32}]
	Using the Theorem \ref{theorem_cons} we have
	\begin{align*}
		\Vert {\mathbfcal P}_{\hat{\pmb\psi}_i^{[\pmb\alpha]} }- {\mathbfcal P}_{{\pmb\psi}_i} \Vert_{\mathcal{L}(\mathbb{H})}^2&=\sup_{\pmb{x}:\Vert\pmb{x}\Vert_\mathbb{H}=1}
		\Vert {\mathbfcal P}_{\hat{\pmb\psi}_i^{[\pmb\alpha]} }\pmb{x}- {\mathbfcal P}_{{\pmb\psi}_i}\pmb{x} \Vert_{\mathbb{H}}^2\\
		&=\sup_{\pmb{x}:\Vert\pmb{x}\Vert_\mathbb{H}=1}\left(
		\Vert {\mathbfcal P}_{\hat{\pmb\psi}_i^{[\pmb\alpha]} }\pmb{x}\Vert_{\mathbb{H}}^2+\Vert  {\mathbfcal P}_{{\pmb\psi}_i} \pmb{x}\Vert_{\mathbb{H}}^2-2\langle {\mathbfcal P}_{\hat{\pmb\psi}_i^{[\pmb\alpha]} } \pmb{x}, {\mathbfcal P}_{{\pmb\psi}_i} \pmb{x}\rangle_{\mathbb{H}}\right)\\
		&\leq 
		2-2_{}	\dfrac{\langle \hat{\pmb\psi}_i^{[\pmb\alpha]} , {\pmb\psi}_i \rangle_{\mathbb{H}}}{\Vert \hat{\pmb\psi}_i^{[\pmb\alpha]}\Vert_{\mathbb{H}}^2\Vert  {\pmb\psi}_i \Vert_{\mathbb{H}}^2}
		=2-2
		\langle \tilde{\pmb\psi}_i^{[\pmb\alpha]} , {\pmb\psi}_i \rangle_{\mathbb{H}}\to 0.
	\end{align*}	
\end{proof}
\newpage

\section{Supplementary Figures}

\begin{figure}[!th] 
    \centering
    \includegraphics[width=.9\textwidth]{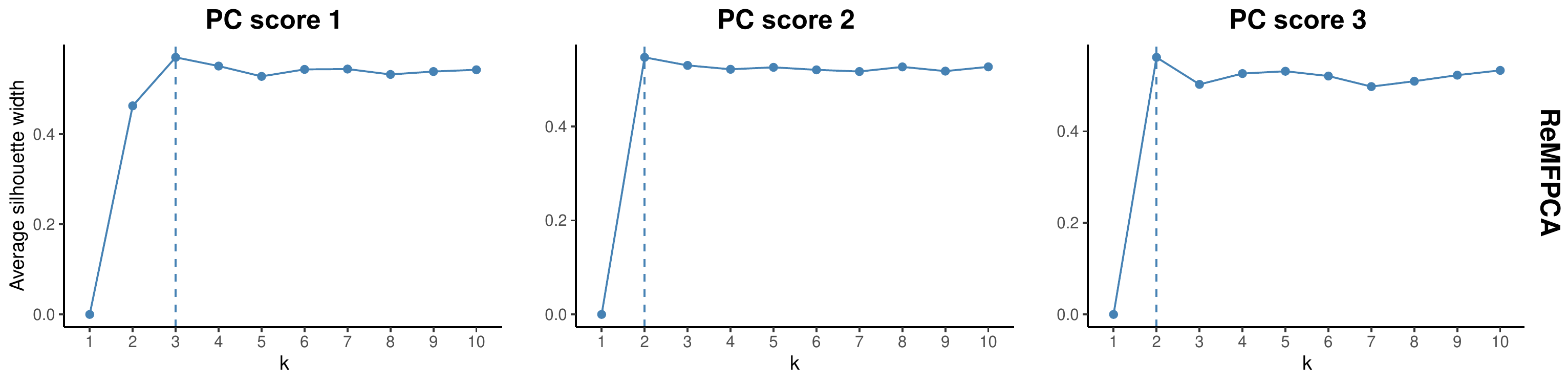}
    \includegraphics[width=.9\textwidth]{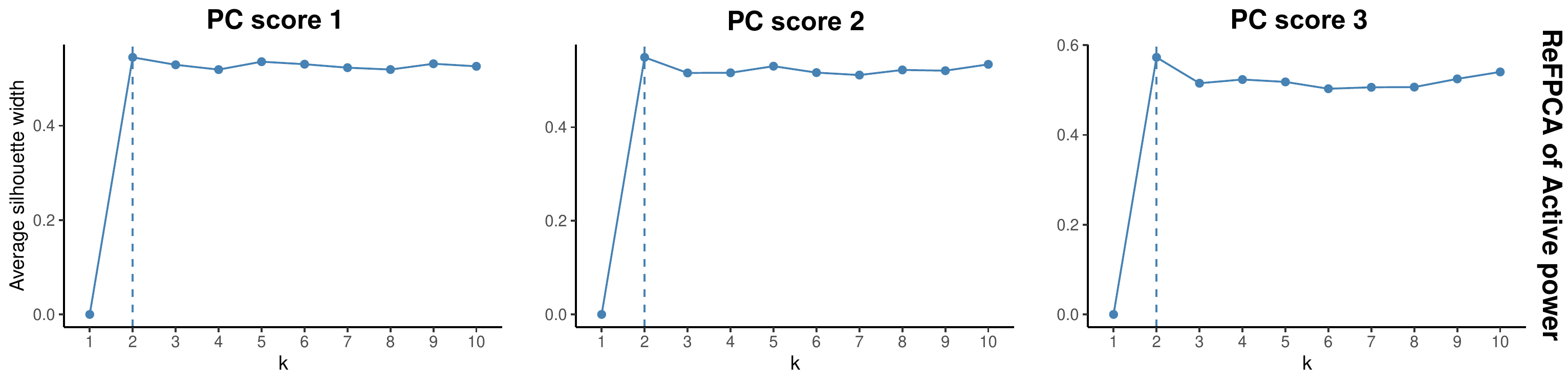}
    \includegraphics[width=.9\textwidth]{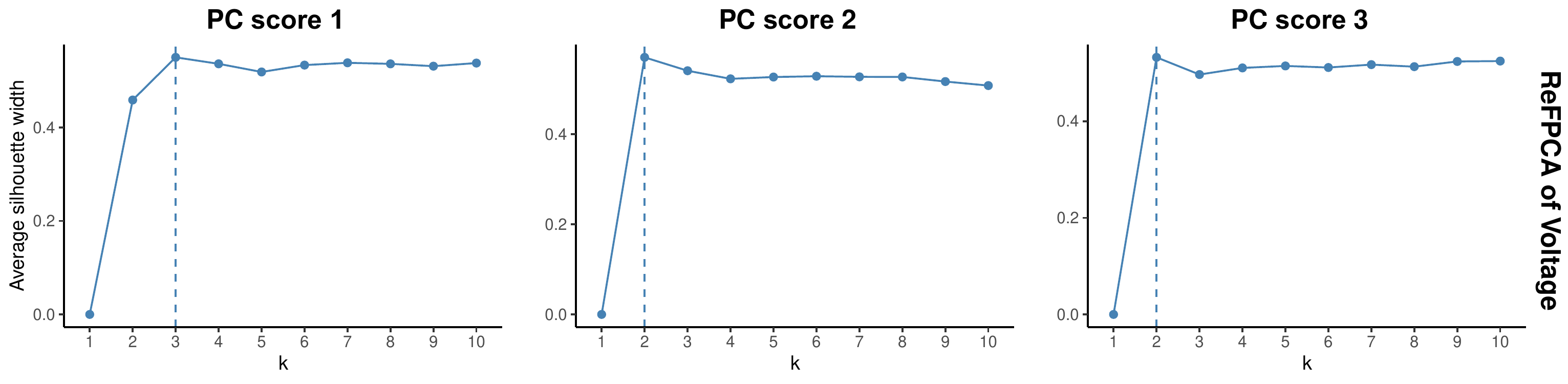}
  \caption{Silhouette score plots.}
  \label{silhouette}
\end{figure}

\begin{figure}[!bh] 
\centering
  \includegraphics[width=.87\textwidth]{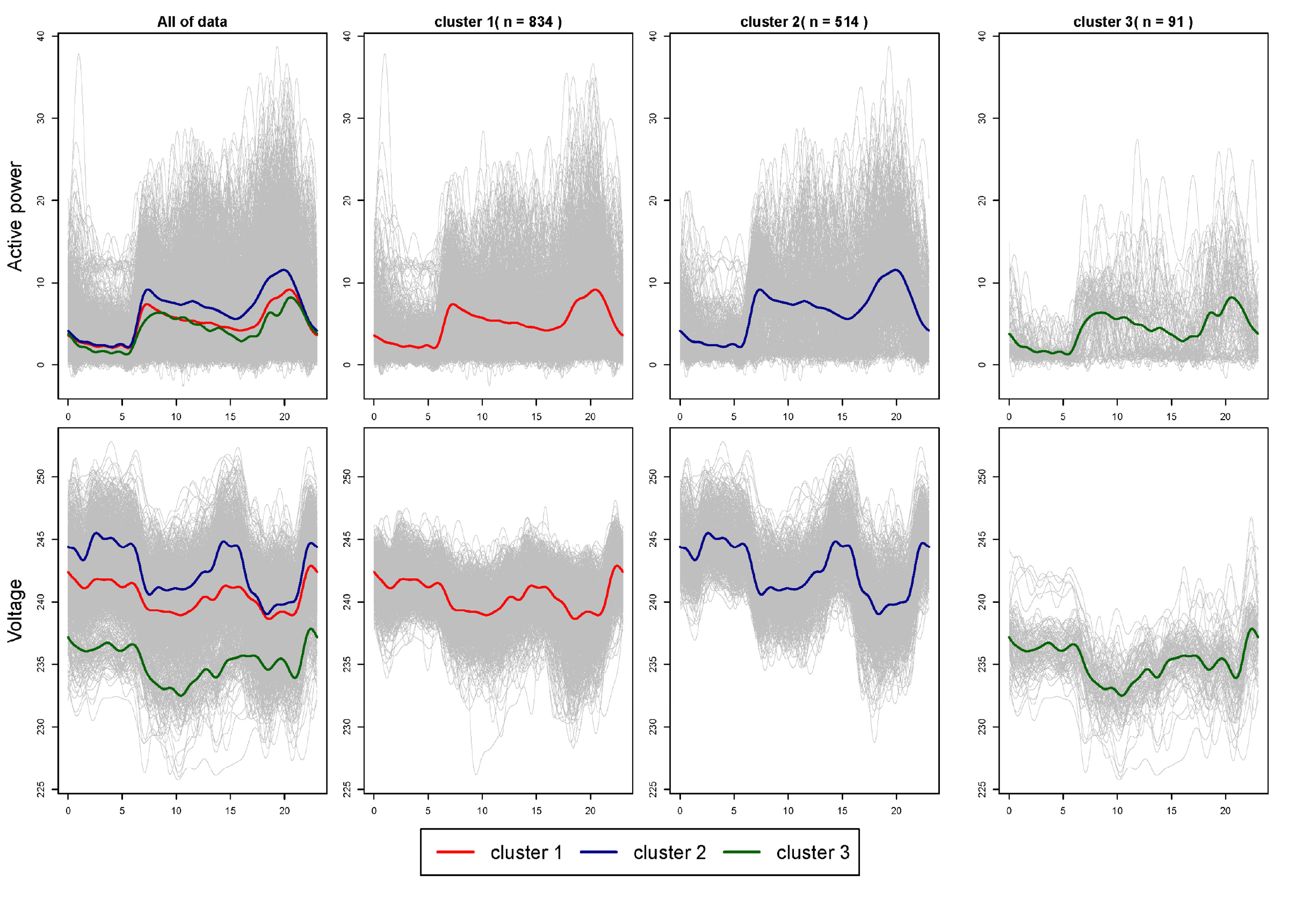}
  \caption{Clustering results based on first PC of ReMFPCA with estimated cluster mean curve.}
  \label{FPC1: cluster}
\end{figure}


\begin{figure}[H] 
  \begin{subfigure}[b]{\textwidth}
    \centering
    \includegraphics[width=.9\textwidth]{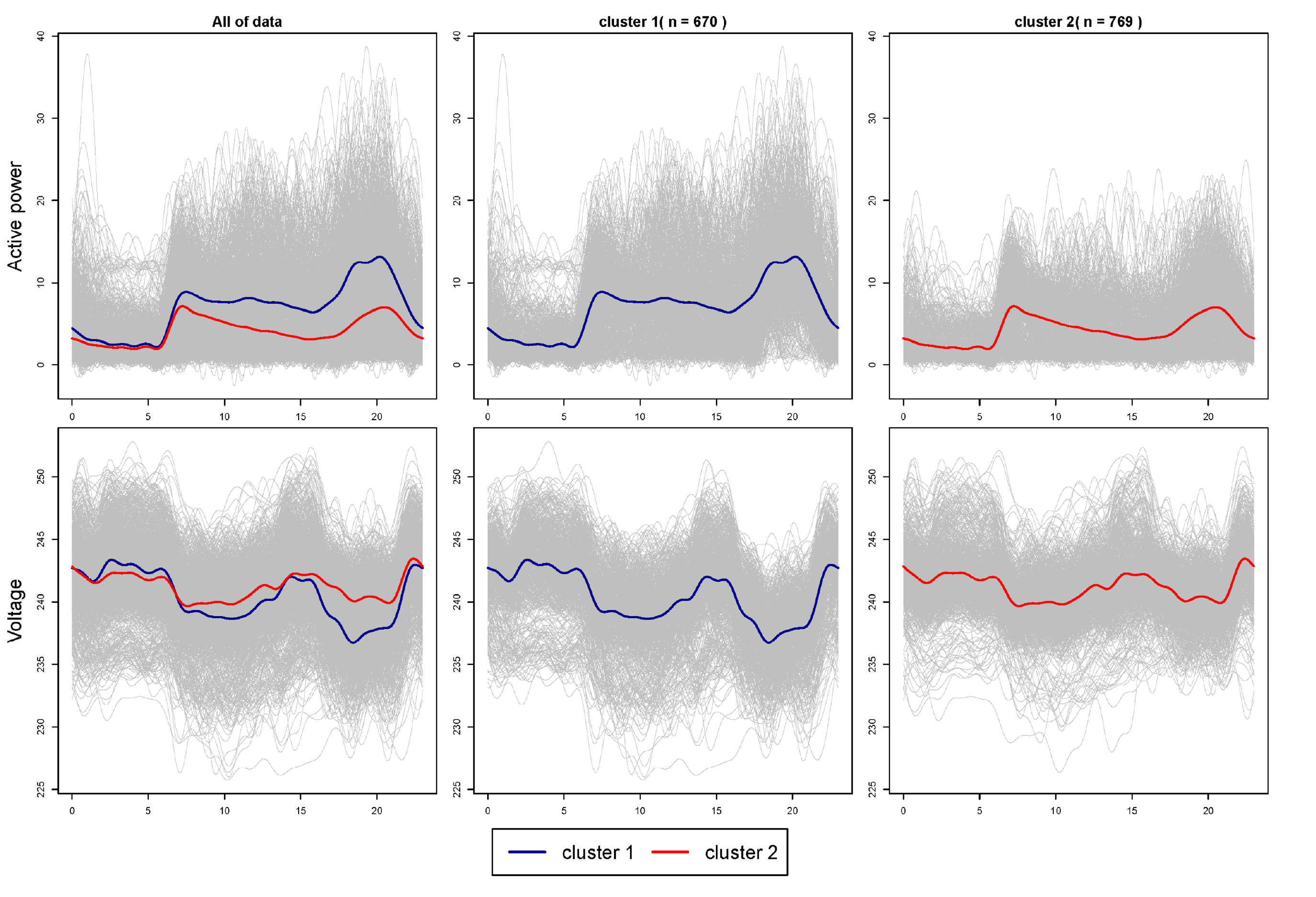}
    \caption{Estimated cluster mean curves}
    \label{FPC2 cluster:plot1}
  \end{subfigure}
  \begin{subfigure}[b]{\textwidth}
    \centering
    \includegraphics[width=.8\textwidth]{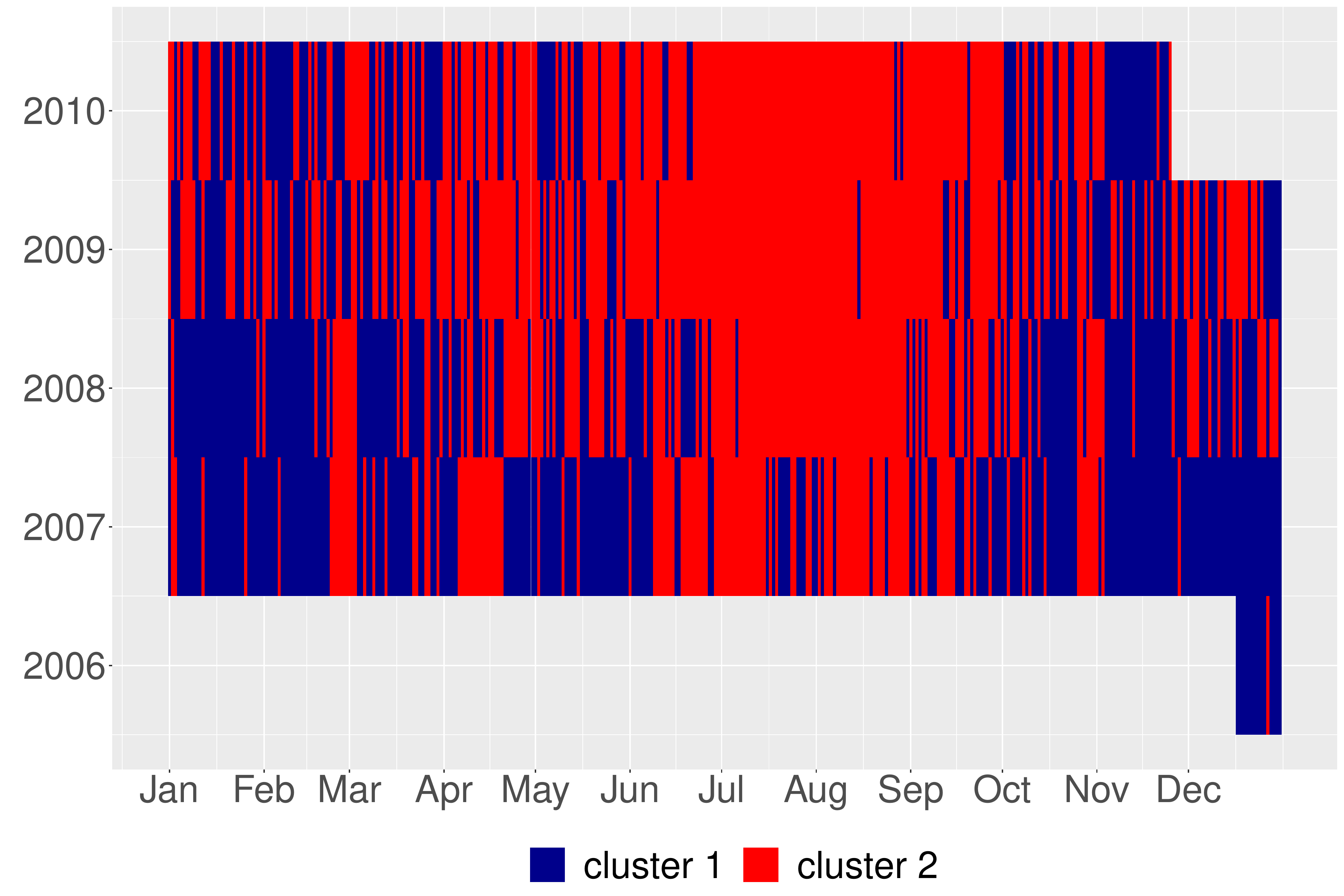}
    \caption{Clustering heatmap}
    \label{FPC2 cluster:plot2}
  \end{subfigure}
  \caption{Clustering results based on 2nd PC of ReMFPCA.}
  \label{FPC2 cluster}
\end{figure}

\bibliographystyle{apalike}	
\bibliography{Mybib}
\end{document}